\documentclass{article}
\usepackage{verbatim}
\usepackage{times}
\usepackage[pdftex]{graphicx}
\usepackage{amssymb,amsfonts,amsmath,amsthm}
\usepackage{float}
\usepackage{xcolor}
\usepackage{enumitem}
\usepackage{bm}
\usepackage{bbm}
\usepackage{dsfont}
\usepackage{url}
\usepackage{caption}
\usepackage{subcaption}
\usepackage{booktabs}
\usepackage{microtype}

\usepackage{titlesec}
\DeclareMathOperator*{\esssup}{ess\,sup}

\DeclareMathOperator*{\argmax}{arg\,max}

\newtheorem{theorem}{Theorem}

\newtheorem{lemma}{Lemma}
\newtheorem{ass}{Assumption}
\newtheorem{definition}{Definition}
\newtheorem{remark}{Remark}
\newtheorem{cor}{Corollary}

\newcommand\EE {\mathbb E}
\newcommand\FF {\mathbb F}
\newcommand\NN {\mathbb N}
\newcommand\RR {\mathbb R}
\newcommand\PP {\mathbb P}

\newcommand\WW {\mathbb W}

\def\qed{\hskip6pt\vrule height6pt width5pt depth1pt}

\title{Contracting a Crowd of Heterogeneous Agents\thanks{E.~Bayraktar is partially supported by NSF grant DMS-2507940 and the Susan M. Smith Chair. I.~Ekren is partially supported by NSF grant DMS-2406240.}}
\author{Guillermo Alonso Alvarez\thanks{Department of Mathematics, University of Michigan. \protect\url{guialv@umich.edu}} \and Erhan Bayraktar\thanks{Department of Mathematics, University of Michigan. \protect\url{erhan@umich.edu}} \and Ibrahim Ekren\thanks{Department of Mathematics, University of Michigan. \protect\url{iekren@umich.edu}}}

\begin{document}
	\maketitle
\begin{abstract}

We study optimal contract design for large populations of heterogeneous agents whose actions generate network spillovers represented by an interaction function. In a linear-quadratic framework, we solve the finite-agent problem and its continuum limit, obtaining explicit optimal contracts and equilibrium efforts. We show that the continuum contract can be evaluated on a large finite sample of agents to obtain admissible contracts that achieve the finite-agent principal's value up to an error of order \(1/N\). This provides a scalable approximation for settings with many interacting agents. We also prove stability with respect to perturbations of the interaction function and provide comparative statics and numerical examples showing how network position affects effort, incentives, and the principal's value. The results identify how optimal incentives should be targeted toward agents whose actions generate larger spillovers.

\end{abstract}

\section{Introduction}

Many contracting problems require a principal to incentivize a large population of agents whose efforts are interdependent. Examples include firms that reward workers in teams, platforms that manage interconnected users, and supply chains where incentives must be coordinated across buyers, suppliers, and intermediaries. In these environments, one agent's effort may affect not only that agent's own output but also the productivity of other agents. These spillovers are heterogeneous: some agents are more influential than others, and the value of incentivizing an agent depends on where that agent is located in the interaction structure.

This paper studies optimal contracting in such environments. We consider a continuous-time principal-agent model in which agents interact through an interaction function $G:[0,1]^2\to \RR$, where $G(u,v)$ measures how the output of type $u$ is affected by effort generated by source type $v$. In a linear-quadratic specification, we solve both the finite $N$-agent model and its continuum-agent counterpart explicitly. The continuum model provides a tractable description of large contracting problems while preserving heterogeneity in network position. We prove that equilibrium efforts, optimal contracts, and the principal's value converge as $N\rightarrow\infty$. Moreover, the continuum contract induces an admissible finite-agent contract that is $\mathcal{O}(1/N)$-optimal.

The starting point is the continuous-time moral-hazard literature. Holmström and Milgrom \cite{holmstrom1987aggregation} characterize optimal incentives in a tractable setting where a single agent controls the drift of output. Sannikov \cite{sannikov2008continuous} introduced the continuation-value approach, and Cvitanić, Possamaï, and Touzi \cite{cvitanic2018dynamic} developed a general stochastic-control formulation of continuous-time principal-agent problems. We retain this framework but replace the single-agent environment with a large population connected by heterogeneous spillovers.

This interaction structure changes the economics of optimal incentives. The principal must account not only for an agent's own response to compensation but also for the way induced effort propagates through the network. As a result, optimal contracts depend on both direct productivity and outward influence. We analyze how optimal contracts use steeper compensation to incentivize agents with greater outward influence. We also establish incentive comparative statics with respect to influence; see Theorem~\ref{Lemma.qualitative} and Theorem~\ref{Theorem.vp.influence}. Moreover, we show stability of both equilibrium efforts and contracts under perturbations of the interaction function; see Theorem~\ref{lemma.Q.stability} and Theorem~\ref{theorem.contracts.stability}.

Our analysis also builds on mean-field game methods. Mean-field games approximate large strategic systems through the interaction between a representative agent and the population distribution. In continuous time, this approach was introduced independently by Lasry and Lions \cite{lasrylions2,lasrylions3,lasrylions1} and Huang, Malhamé, and Caines \cite{CainesMeanField,CainesMeanfield2}. The principal-agent literature with mean-field interactions begins with Élie, Mastrolia, and Possamaï \cite{manymanyagents}, who study a moral-hazard problem with infinitely many homogeneous agents. Their weak-formulation approach is close to Carmona and Lacker \cite{weakformulationmeanfieldgames}; the resulting principal's problem is then handled through mean-field control and dynamic programming, as in \cite{bellmanphamwei} and \cite{bayraktarcosso}. Related extensions include \cite{cvitanicbayraktarzhang}, \cite{possamaienergy}, \cite{bayraktaryuchongzhang}, \cite{djetestackelberg}, \cite{gockelauriere}, \cite{agencymastrolia}, and \cite{rankbasedfieldofagents}.

The main difference between our model and this literature is the form of interaction. Existing mean-field principal-agent models typically aggregate the population through a common distributional state. Here, the effect of one agent's effort on another agent's output depends on their positions in the interaction function. This makes the principal's problem sensitive to the source of effort: incentivizing an influential type can create value that is realized elsewhere in the population.

To capture this structure, we use ideas from mean-field games with interaction functions. Parise and Ozdaglar \cite{parise} and Caines and Huang \cite{CainesGraphon} introduced models in which continuum agents interact through a function $G:[0,1]^2\to \RR$ describing population connectivity. This line of work includes \cite{StochasticGraphonGamesCarmona}, \cite{wubayraktaruniformgraphon}, \cite{graphonpropagationerhan}, \cite{graphonsystems}, \cite{pariseozaglar2}, \cite{donhambayraktar}, \cite{hongyibayraktar}, \cite{phamgraphon}, \cite{phamgraphon2}, and \cite{donhanhebayraktar}. Such models allow agents to differ not only in their states but also in how they are connected to other parts of the population.

The heterogeneous continuum formulation creates technical difficulties. Working directly with uncountably many state and noise processes is delicate, and standard mean-field arguments are not immediately applicable. We follow the labeling idea of Lacker and Soret \cite{labelgraphon}: the agent's type is included as part of the state, and the population distribution records the joint law of type and output. This converts the heterogeneous system into a mean-field problem on the joint state-type space.

In the linear-quadratic model, we compute the optimal contracts and equilibrium effort profiles in both the finite and continuum problems; see Theorem~\ref{opt.contract.finite.example1} and Theorem~\ref{opt.contract.interaction.example.1}. We then prove three convergence results:
\begin{itemize}
	\item[$i)$] the agents' Nash equilibrium efforts converge as $N\rightarrow\infty$; see Theorem~\ref{lemmaQ};
	\item[$ii)$] the optimal contracts and the principal's value converge as $N\rightarrow\infty$; see Theorem~\ref{conv.Value};
	\item[$iii)$] the continuum optimal contract yields a finite-agent contract that satisfies the agents' reservation constraints, induces a unique Nash equilibrium effort profile, and is $\mathcal{O}(1/N)$-optimal for the principal; see Theorem~\ref{nearoptimal}.
\end{itemize}

At the methodological level, we adapt the homogeneous-agent approach of Élie, Mastrolia, and Possamaï \cite{manymanyagents} to a setting with heterogeneous interactions. Since their argument relies on homogeneity, we combine it with the labeling formulation of Lacker and Soret \cite{labelgraphon}. The resulting state variable is the joint distribution of output and type. We derive the associated Wasserstein Hamilton-Jacobi-Bellman equation and a verification lemma, Lemma~\ref{thm.verification}, which yields the optimal continuum contract. The verification result is stated in a general form and can be applied beyond the linear-quadratic example studied in detail here.

The remainder of the paper is organized as follows. Section~\ref{section.technicalsetup} develops the continuum contracting formulation and the verification lemma. Section~\ref{section.last} studies the linear-quadratic model. Subsection~\ref{example1.finite.agents} solves the $N$-agent problem, while Subsection~\ref{continuum.agents.example} solves the continuum problem using the labeling method and Lemma~\ref{thm.verification}. We then provide comparative statics and numerical simulations showing how connectivity affects the principal's value, agents' efforts, and optimal contracts. We also prove stability with respect to perturbations of the interaction function. Finally, Subsection~\ref{subsectionconvergence} establishes convergence from the finite-agent problem to the continuum-agent problem as $N\rightarrow\infty$. Readers primarily interested in the explicit linear-quadratic model may go directly to Section~\ref{section.last}; Section~\ref{section.technicalsetup} contains the verification machinery used to justify the continuum formulation.
\subsection{Notation}
Throughout the paper, we use the following notation. For $N\in \mathbb{N}$, define $[N]:=\{1,\ldots,N\}$. We denote by $\mathcal{C}:=C([0,T],\RR)$ the space of real-valued continuous functions endowed with the norm $\|\cdot\|_{\infty}$, where $\|x\|_{\infty} :=\sup_{0\leq t \leq T}|x(t)|$ for all $x\in \mathcal{C}$, and $|\cdot|$ denotes the usual Euclidean norm.

We denote by $\mathcal{P}^2(S)$ the set of probability measures on a normed space $(S,\|\cdot \|_{S})$ with finite second moments, and by $\mathcal{P}_{\text{Unif}}^2([0,1]\times\RR)$ the set of probability measures on $[0,1]\times \RR$ whose first marginal is uniform on $[0,1]$. For $\nu, \eta \in \mathcal{P}^2(S)$, define the Wasserstein-2 distance between $\nu$ and $\eta$ by
\begin{align*}    \mathcal{W}_2\left(\eta,\nu\right) := \inf_{\kappa \in\Lambda(\nu,\eta)}\left(\int_{S\times S}\|x-y\|_S^2\kappa(dx,dy)\right)^{1/2},
\end{align*}
where $\Lambda(\nu,\eta)$ denotes the set of measures $\kappa\in \mathcal{P}^2(S\times S)$ such that $\kappa(dx,S) = \nu$ and $\kappa(S,dy) = \eta$. Moreover, for all $\mu \in \mathcal{P}^2(S)$, define $\|\mu\|_2 := \left(\int_{S}\|x\|^2_S\mu(dx) \right)^{1/2}$.

	\section{Verification lemma for the principal-agent problem}\label{section.technicalsetup}
       In this section, we develop a general procedure for contracting problems with a continuum of heterogeneous agents. The method provides the tools needed to solve the model introduced in Section \ref{section.last}. Inspired by Lacker and Soret \cite{labelgraphon}, we describe the agents' stochastic game as a standard mean-field game involving a representative agent and a population distribution $\mu\in \mathcal{P}^2([0,1]\times \mathcal{C})$ that records the joint distribution of agent type and output. We then apply the method introduced in Élie et al. \cite{manymanyagents} to reduce the principal's problem to a mean-field control problem.
   \subsection{The principal-agent problem}

    Let $\Omega := [0,1]\times \RR \times \mathcal{C}$ be the canonical space. We denote by $(U,\Psi,W)$ the canonical random element in $\Omega$, representing the representative agent's type, initial output, and idiosyncratic noise, respectively. Let $\lambda \in \mathcal{P}^2_{\text{Unif}}\left([0,1]\times \RR\right)$ be the initial joint distribution of the agent's type and initial output. We define $\PP_0:= \lambda\otimes \WW \in \mathcal{P}^2(\Omega)$, where $\WW$ is the Wiener measure on $\mathcal{C}$, and denote by $\FF:= (\mathcal{F}_t)_{0 \leq t\leq T}$ the $\PP_0$-augmentation of the filtration $\sigma\Big(U,\Psi,\left\{W_s, 0 \leq s\leq t\right \}\Big)_{t\geq 0}$. In addition, we introduce the following classes of processes on $(\Omega,\mathcal{F}_T,\PP_0)$:
	\begin{align*}
		\mathbb{H}([0,T]) &:=\bigg\{ Z, \FF-\text{predictable } \bigg|\text{ }\mathbb{E}^{\PP_0}\left[\left(\int_0^T|Z_s|^2ds\right)^{p/2}\right]<\infty, \forall p\geq 0\bigg\}, \\
    \mathbb{H}_{\text{exp}}([0,T])&:= \left\{ Z,\FF-\text{predictable }\bigg| \text{ }\EE^{\PP_0}\left[\exp\left(p\int_0^T\big|Z_s\big|^2ds \right)\right]<\infty, \forall p\geq 0  \right\},\\
    \mathbb{D}_{\text{exp}}([0,T])&:= \left\{ Y,\FF-\text{adapted, càdlag }\bigg| \text{ }\EE^{\PP_0}\left[\exp\left(p\sup_{0\leq s\leq T}\big|Y_s\big| \right)\right]<\infty, \forall p\geq 0  \right\},\\
    \mathbb{L}_{\text{exp}}\left( \mathcal{F}_T\right)&:= \left\{ \xi, \mathcal{F}_T-\text{measurable } \bigg| \text{ }\EE^{\PP_0}\left[\exp\left(p|\xi|  \right)\right]<\infty, \forall p\geq 0  \right\}.
	\end{align*}

	We start by describing the dynamics of the output process $X$, which represents the project controlled by the representative agent. Let $\sigma: [0,T]\times [0,1]\times \mathcal{C} \to \RR$ be a progressively measurable map. We assume that the following stochastic differential equation admits a unique $\FF$-adapted strong solution:
	\begin{align*}
		X_t = \Psi+ \int_0^t\sigma_s(U,X)dW_s, \quad \PP_0 -a.s., \quad 0 \leq t \leq T.	\end{align*}
	We denote by $\alpha$ the representative agent's effort process, assumed to be real-valued and $\FF$-adapted. Next, we introduce the measurable mapping $b: [0,T]\times [0,1]\times\mathcal{C}  \times \mathcal{P}^2\left([0,1]\times\mathcal{C}\right)  \times \RR \to \RR$ satisfying the following condition.
	\begin{ass}\label{ass1}
		The mappings $b$ and $\sigma$ satisfy the following properties:
		\begin{itemize}
            \item[(a)] $0<\underline{\sigma}\leq \sigma_t(u,x)\leq \overline\sigma <\infty$, for all $(t,u,x) \in [0,T]\times[0,1]\times\mathcal{C}$.
			\item[(b)] The process $(t,u,x)\to b_t(u,x,\mu,a)$ is $\FF$-progressively measurable, for all $(\mu, a)\in \mathcal{P}^2([0,1]\times \mathcal{C})\times \RR$.
            \item[(c)] There exists a constant $C>0$ such that
			$$\big|b_t(u,x,\mu,a)\big| \leq C \left(1+\|x\|_{\infty}+|a|+\int_{[0,1]\times    \mathcal{C}}\|x\|_{\infty}\mu(du,dx) \right).$$
		\end{itemize}
	\end{ass}
    We denote by $\mathcal{A}$ the set of $\FF$-adapted effort processes $\alpha$ for which there exists $\epsilon>0$ such that
	\begin{align*}    &\mathbb{E}^{\mathbb{P}_0}\left[\mathcal{E}\left(\int_0^{\cdot}\sigma_s(U,X)^{-1}b_s(U,X,\mu,\alpha_s)dW_s \right)_T^{1+\epsilon} \right]+\EE^{\PP_0}\left[\exp\left(p\int_0^T|\alpha_s|^2ds  \right) \right]< \infty,
	\end{align*}
	 for all $p \geq 0$, and $\mu \in \mathcal{P}^2([0,1]\times \mathcal{C})$.

	Therefore, given $(\alpha,\mu)\in \mathcal{A} \times \mathcal{P}^2([0,1]\times \mathcal{C})$, we define the probability measure $\PP^{\alpha,\mu}$ determined by the following Radon-Nikodym derivative.
	\begin{equation*}
		\frac{d\PP^{\alpha,\mu}}{d\PP_0} := \mathcal{E}\left(\int_0^{\cdot}\sigma_s(U,X)^{-1} b_s(U,X,\mu,\alpha_s)dW_s\right)_T.
	\end{equation*}
	Using Girsanov's theorem, the process
	\begin{align*}    W_t^{\alpha,\mu}:= W_t -\int_0^t\sigma_s(U,X)^{-1}b_s(U,X,\mu,\alpha_s)ds,
	\end{align*}
	is a $\PP^{ \alpha,\mu}$-Brownian motion. Therefore, for all $(\alpha,\mu ) \in \mathcal{A}\times \mathcal{P}^2([0,1]\times \mathcal{C})$, the output process has the following dynamics under $\PP^{\alpha,\mu}$:
	\begin{align*}
		X_t = \Psi+\int_0^tb_s(U,X,\mu,\alpha_s)ds + \int_0^t\sigma_s(U,X)dW^{\alpha,\mu}_s,  \quad \PP^{\alpha,\mu}-a.s., \quad 0\leq t \leq T.
	\end{align*}
	\begin{remark}
		Following the formulation introduced by {Lacker and Soret} \cite{labelgraphon}, we encode the model's heterogeneity through an additional state variable $U$ labeling agent type. The main advantage of this approach is that it avoids the technical difficulties arising from uncountably many Brownian motions. Moreover, it allows us to express the agents' problem as a standard mean-field game played by a representative agent.
	\end{remark}
	In this setting, the principal compensates the representative agent with a lump sum payment $\xi$ at time $T$ and a continuous payment stream $\chi : [0, T]\times \Omega \to \RR $. We assume that $\chi$ is an $\FF$-adapted process and $\xi$ is an $\mathcal{F}_T$-measurable random variable. With a slight abuse of notation, we denote by $\mathcal{F}_T$ the set of $\mathcal{F}_T$-measurable random variables and by $\FF$ the set of $\FF$-adapted processes.
    \begin{remark}
    Note that for all $\xi \in \mathcal{F}_T$, there exists a Borel-measurable function $\tilde{F}:[0,1]\times \mathcal{C}\to \RR$ such that $\xi = \tilde{F}(U,X)$, $\PP_0-a.s.$. Therefore, for almost every $u \in [0,1]$, the random variable $\tilde{F}(u,X)$ represents the contract assigned to type $u$. The same analysis applies to the continuous incentive $\chi$ and the effort process $\alpha$.
    \end{remark}
    Next, we introduce the following mappings describing the agents' characteristics:
    \begin{itemize}
        \item Cost of effort: $c:[0,T]\times [0,1]\times \mathcal{C}\times \mathcal{P}^2([0,1]\times \mathcal{C})\times \RR \to \RR$.
        \item Discounting process: $k:[0,T]\times [0,1]\times \mathcal{C}\times \mathcal{P}^2([0,1]\times \mathcal{C})\times \RR\to \RR$.
        \item Terminal utility: $U_a:\RR \to \RR$.
        \item Continuous utility: $u_a :[0,T]\times[0,1]\times \mathcal{C}\times \mathcal{P}^2\left([0,1]\times \mathcal{C}\right)\times \RR\to \RR$.
    \end{itemize}
    These mappings satisfy the following assumption.
    \begin{ass}\label{ass3}
The mappings $c$, $k$, $U_a$, and $u_a$ satisfy the following properties:
\begin{itemize}
\item[(a)] $c$, $k$, $U_a$, and $u_a$ are measurable.
\item[(b)] The mappings $(t,u,x)\to c(t,u,x,\mu,a)$,  $(t,u,x) \to k(t,u,x,\mu,a)$, and $(t,u,x) \to u_a(t,u,x,\mu,\chi)$ are $\FF$-progressively measurable for all $(\mu,a,\chi) \in  \mathcal{P}^2([0,1]\times \mathcal{C})\times\RR^2$.
\item[(c)] There exists a constant $K>0$ such that
\begin{equation*}\label{cond.utility}
|u_a(t,u,x,\mu,\chi)|\leq K\left(1+\|x\|_{\infty}+\int_{[0,1]\times \mathcal{C}}\|x\|_{\infty}\mu(du,dx)+|\chi|\right).
\end{equation*}
for all $(t,u,x,\mu,\chi)\in [0,T]\times [0,1]\times \mathcal{C}\times \mathcal{P}^2([0,1]\times \mathcal{C})\times \RR$.
\item[(d)] $U_a$ is invertible.
\item[(e)]  There exists a constant $K>0$ such that
\begin{align*}
   \left|c(t,u,x,\mu,a)\right| \leq K\left(1+|a|^2+\|x\|_{\infty}+\int_{[0,1]\times\mathcal{C} }\|x\|_{\infty}\mu(du,dx)\right),
\end{align*}
 for all $(t,u,x,\mu,a) \in [0,T]\times[0,1]\times \mathcal{C}\times \mathcal{P}^2([0,1]\times\mathcal{C})\times \RR$.
\item[(f)] $k$ is bounded.
\end{itemize}
\end{ass}
    Given a contract $(\chi,\xi) \in  \FF\times \mathcal{F}_T$ satisfying the integrability requirements in Definition \ref{def.admissible.contracts} below, the representative agent maximizes the following objective over $0 \leq t \leq T$.
	\begin{align*}   J^a_t(\alpha,\mu,\chi,\xi) &:= \EE^{\PP^{\alpha,\mu}}\left[\mathcal{K}_{t,T}^{\alpha}U_a\left(\xi\right)+\int_t^T\mathcal{K}_{t,s}^{\alpha}\big(u_a(s,U,X,\mu,\chi_s)-c(s,U,X,\mu,\alpha_s)\big)ds\bigg|\mathcal{F}_t\right],
	\end{align*}
	where
	\begin{equation*}   \mathcal{K}_{t,s}^{\alpha} := \exp\left(\int_{t}^s k_r\big(U,X,\mu,\alpha_r\big)dr\right), \quad t\leq s\leq T.
	\end{equation*}
	Furthermore, for a contract $(\chi,\xi)\in \FF\times \mathcal{F}_T$ and a population distribution $\mu\in \mathcal{P}^2([0,1]\times \mathcal{C})$, we introduce the agent's continuation value:
	$$	V^{a}_t(\mu,\chi,\xi):= \esssup_{\alpha \in \mathcal{A}}J_t^a (\alpha,\mu,\chi,\xi).
	$$

	\begin{definition}
		Let $(\chi,\xi) \in \FF\times \mathcal{F}_T$ be a contract. We say that a pair  $(\alpha,\mu)\in \mathcal{A}\times \mathcal{P}^2([0,1]\times \mathcal{C})$ is a heterogeneous mean-field equilibrium given $(\chi,{\xi})\in \FF\times \mathcal{F}_T$ if
		\begin{align*}   &J_0^a(\alpha,\mu,\chi,\xi) = V^a_0(\mu,\chi,\xi),\quad \lambda-a.s. \\ &\PP^{\alpha,\mu}\circ (U,X)^{-1} = {\mu}.
		\end{align*}
		Moreover, we denote by $\mathcal{M}^*({\chi},\xi)$ the set of heterogeneous mean-field equilibria given the contract $(\chi,\xi)$.
	\end{definition}
       In this work, we consider contracts satisfying the following additional properties.
	\begin{definition}\label{def.admissible.contracts}
		We denote by $\Sigma$ the set of contracts $(\chi,\xi)\in \FF\times \mathcal{F}_T$ such that
        \begin{itemize}
            \item[a)] $\mathcal{M}^*(\chi,\xi)\neq \emptyset.$
            \item[b)] $\EE^{\PP_0}\left[\exp\left(p|U_a(\xi)|+p|\xi|+p\int_0^T|\chi_t|dt \right) \right]<\infty,\quad$ for all $p \geq 0$.
        \end{itemize}
	\end{definition}
		Additionally, we assume that the contracts satisfy the agent's reservation utility constraint:
    \begin{align}\label{reservation}
		J^a_0(\alpha,\mu,\chi,\xi) \geq \tilde{R}_a, \quad  \lambda -a.s., \quad \forall (\alpha,\mu) \in \mathcal{M}^*(\chi,\xi),
	\end{align}
	where $\tilde{R}_a \in \mathcal{F}_0$ represents the reservation utility of the representative agent. We denote by $\Sigma_a\subset \Sigma$ the set of contracts satisfying \eqref{reservation}.
    \begin{remark}
        Note that $\tilde{R}_a\in \mathcal{F}_0$ implies the existence of a measurable mapping ${R}_a : [0,1]\times \RR\to \RR$ such that $\tilde{R}_a = {R}_a(U,\Psi)$, $\lambda-a.s.$
    \end{remark}

  The principal, in turn, maximizes the following objective:
  \begin{equation}\label{principal.obj}
		J_p(\chi,\xi) := \sup_{({\alpha},\mu)\in \mathcal{M}^*({\chi},\xi)}\EE^{\PP^{\alpha,\mu}}\left[U_p\bigg(X_T-\xi-\int_{0}^T\chi_s ds\bigg)\right],\quad
		V_p := \sup_{(\chi,\xi) \in {\Sigma_a}}J_p(\chi,\xi),
	\end{equation}
	where  $U_p: \RR \to \RR$ denotes the principal's utility function. We assume that $U_p$ is measurable and has polynomial growth. The additional integrability condition on $\xi$ in Definition \ref{def.admissible.contracts} ensures that the principal's objective is well defined.
	\begin{remark}
		Our formulation is well suited to studying contracting problems with an infinite population of heterogeneously interacting agents. Interactions can enter through the drift of $X$, the objectives of the principal and agent, or the reservation utilities.
	\end{remark}
    \subsection{Connection to mean-field control}

The principal-agent problem described in equation \eqref{principal.obj} has the same structure as in Élie et al. \cite{manymanyagents}, except that we use an extended state space. Thus, we follow the same approach to reduce the principal's problem to a mean-field control problem. The main idea is to conveniently parameterize the set of admissible contracts $\Sigma_a$ using backward stochastic differential equations (BSDEs). This parameterization allows us to characterize, given a contract $(\chi,\xi)\in \Sigma_a$, the optimal effort of the representative agent and the corresponding continuation utility. This, in turn, allows us to address the principal's problem using a dynamic programming approach.

 First, we define the agent's Hamiltonian as follows.
	\begin{align*}   h_t(u,x,y,{\mu},z,\chi,a) &:= z b_t(u,x,\mu,a)+u_a(t,u,x,\mu,\chi)-c_t(u,x,\mu,a)+k_t(u,x,\mu,a) y, \\   H_t(u,x,y,{\mu},z,\chi) &:= \sup_{a \in \RR} h_t(u,x,y,\mu,z,\chi,a),
	\end{align*}
where $h: [0,T]\times [0,1]\times\mathcal{C}\times \RR \times\mathcal{P}^2\left([0,1]\times \mathcal{C}\right)\times \RR^3\to \RR $.

\begin{ass}\label{ass.max.hamiltonian}
    There exists a unique measurable mapping $\hat{\alpha} :[0,T]\times [0,1]\times \mathcal{C}\times \RR\times \mathcal{P}^2([0,1]\times \mathcal{C})\times \RR^2\to \RR$ with linear growth satisfying
    \begin{align*}
&h_t\left(u,x,y,\mu,z,\chi,\hat{\alpha}_t(u,x,y,\mu,z,\chi)\right) = H_t(u,x,y,\mu,z,\chi),
    \end{align*}
    for all $(t,u,x,y,\mu,z,\chi)\in[0,T]\times[0,1]\times \mathcal{C} \times \RR \times \mathcal{P}^2([0,1]\times \mathcal{C}) \times \RR^2$.
\end{ass}
The following lemma characterizes heterogeneous mean-field equilibria in terms of solutions to a BSDE. It is the labeled-state analogue of the BSDE characterization of Proposition $3.1$ in \cite{manymanyagents}.
	\begin{lemma}\label{BSDE.Lemma}%
		Let $(\chi,\xi) \in \Sigma$. We define the following BSDE
		\begin{align}\label{BSDE.2}
			Y_t &= U_a(\xi) + \int_t^T H_s(U,X,Y_s,\mu,Z_s,\chi_s)ds-\int_t^TZ_s\sigma_s(U,X)dW_s,\quad \PP_0-a.s., \quad t\in [0,T],\\
			\mu &= \PP^{\hat{\alpha}(\cdot),\mu}\circ (U,X)^{-1}. \nonumber
        \end{align}
Suppose that Assumptions \ref{ass1}, \ref{ass3}, and \ref{ass.max.hamiltonian} hold.
\begin{itemize}
\item If equation \eqref{BSDE.2} admits a solution $(Y,Z,\mu)$ with $Z\in \mathbb{H}_{\text{exp}}([0,T])$, $Y\in \mathbb{D}_{\text{exp}}([0,T])$, and $\mu\in \mathcal{P}^2\left([0,1]\times \mathcal{C}\right)$, we have
$(\hat{\alpha}(\cdot),\mu) \in \mathcal{M}^*(\chi,\xi)$, where $\hat{\alpha}$ is defined in Assumption \ref{ass.max.hamiltonian}.
\item Let $(\alpha^*,\mu)\in \mathcal{M}^*(\chi,\xi)$. Then, equation \eqref{BSDE.2} admits a solution $(Y,Z,\mu)$ with $Z\in \mathbb{H}_{\text{exp}}([0,T])$, $Y\in \mathbb{D}_{\text{exp}}([0,T])$, and $\mu\in \mathcal{P}^2\left([0,1]\times\mathcal{C}\right)$ such that $\alpha^*_t = \hat{\alpha}_t(U,X,Y_t,\mu,Z_t,\chi_t)$ $dt\otimes \PP_0-a.s.$
\end{itemize}
\end{lemma}
The proof follows from Proposition $3.1$ in \cite{manymanyagents} after enlarging the state from the output path to the type-output pair. Since \(U\) is \(\mathcal{F}_0\)-measurable and constant in time, it enters the coefficients only as a parameter and does not change the Girsanov argument or the martingale representation. 
    \begin{cor} \label{theorem.mean.field.control}
   Suppose that Assumptions \ref{ass1}, \ref{ass3}, and \ref{ass.max.hamiltonian} hold. Then, the principal's problem can be expressed as the following mean-field control problem.
		\begin{align*}
			&V_p = \sup_{Y_0\geq \tilde{R}_a}\sup_{\chi,Z \in  \mathbb{H}_{\text{exp}}([0,T])}\hat{J}_p(Y_0,Z,\chi), \\
			&\hat{J}_p(Y_0,Z,\chi):=\sup_{(\hat{\alpha}(\cdot),\mu)\in\mathcal{M}^*\left(\chi,U_a^{-1}\big(Y_T^{Y_0,Z,\chi}\big)\right)} \EE^{\PP^{ \hat{\alpha}(\cdot),\mu}}\left[U_p\bigg(X_T-U_a^{-1}\big(Y_T^{Y_0,Z,\chi}\big)-\int_0^T \chi_s ds\bigg) \right], \\
			X_t &=\Psi +\int_0^t b_s\big(U,X,\mu,\hat{\alpha}_s(U,X,Y_s,\mu,Z_s,\chi_s)\big)ds + \int_0^t\sigma_s(U,X)dW^{\hat{\alpha}(\cdot),\mu}_s,\quad \mathbb{P}^{\hat{\alpha}(\cdot),\mu}-a.s.\\
			Y_t &= Y_0 +\int_0^t\Big(b_s\big(U,X,\mu,\hat{\alpha}_s(U,X,Y_s,\mu,Z_s,\chi_s)\big)Z_s-H_s(U,X,Y_s,\mu,Z_s,\chi_s)\Big)ds\\
			&\qquad +\int_0^tZ_s\sigma_s(U,X)dW^{\hat{\alpha}(\cdot),\mu}_s, \\
            \mu &= \PP^{\hat{\alpha}(\cdot),\mu} \circ (U,X)^{-1},
		\end{align*}
        where $\hat{\alpha}:[0,T]\times [0,1]\times \mathcal{C}\times \RR\times \mathcal{P}^2([0,1]\times \mathcal{C})\times \RR^2\to \RR$ is defined in Assumption \ref{ass.max.hamiltonian}.
    \end{cor}
    This corollary adapts Corollary $3.1$ of Élie et al. \cite{manymanyagents}.
    \subsection{Verification lemma}
We continue by adapting the approach of Élie et al. \cite{manymanyagents}. We tackle the mean-field control problem introduced in Corollary \ref{theorem.mean.field.control} using the dynamic programming method proposed by Pham and Wei \cite{bellmanphamwei}. Throughout the verification argument, we identify \([0,1]\times\mathbb R^2\) with a subset of \(\mathbb R^3\). The value function is defined on \(\mathcal P^2(\mathbb R^3)\), and Lions derivatives are understood in the standard sense on \(\mathcal P^2(\mathbb R^3)\). The HJB equation and the verification result are evaluated only on measures supported in \([0,1]\times\mathbb R^2\) whose first marginal is uniform in \([0,1]\).

We first recall the definition of the lifted derivative of a functional in $\mathcal{P}^2( \RR^3)$, introduced by Lions \cite{CollegedeFranceLions}. For $\phi  : \mathcal{P}^2(\RR^3)\to \RR$, define the lifted version of $\phi$ as the mapping $\tilde{\phi} : L^2(\Omega,\mathcal{F}_T,\RR^3)\to \RR$ satisfying
$$\tilde{\phi}\left(\eta\right)=\phi\left(\mathcal{L}\big(\eta\big)\right), \quad \eta \in L^2(\Omega,\mathcal{F}_T,\RR^3).$$
We say that $\phi$ is differentiable in $\mathcal{L}(\eta)$ if $\tilde{\phi}$ is Fréchet differentiable at $\eta$. In this case, we denote by $D\tilde{\phi}_{\eta}:L^2(\Omega,\mathcal{F}_T,\RR^3)\to \RR$ the Fréchet derivative of $\tilde{\phi}$ at $\eta$. By Riesz's Theorem, $D\tilde{\phi}_\eta$ can be represented as
$\Big(\partial_{\nu}\phi\big(\mathcal{L}(\eta)\big)\Big)\big(\eta\big)$ for a suitable function $\partial_{\nu}\phi(\mathcal{L}(\eta)):\RR^3 \to \RR^3$. Thus, we define the lifted derivative of $\phi$ at $\nu = \mathcal{L}(\eta)$ as $\partial_{\nu}\phi(\mathcal{L}(\eta))$. Furthermore, we say that $\phi\in C^2\big(\mathcal{P}^2(\RR^3)\big)$ if
\begin{itemize}
    \item $\phi$ is differentiable at $\nu$ for all $\nu \in \mathcal{P}^2(\RR^3) $.
    \item For all $\nu \in \mathcal{P}^2(\RR^3)$, the function $\partial_{\nu}\phi(\nu)$ is differentiable.
\end{itemize}
We refer the reader to the monograph {\cite{Carmonabook}} for a detailed exposition of the theory.

Next, we fix the stochastic basis $(\Omega,\mathcal{F}_T,\PP)$ and define a $\PP$-Brownian motion $B$ and a random vector $(U,X_0)\sim \lambda$, independent of $B$. We restrict ourselves to the Markovian case: we assume that there exist measurable functions $\hat{b}, \hat{c}:[0,T]\times[0,1]\times \RR \times \RR\times \mathcal{P}([0,1]\times \mathbb{R})\times \RR^2 \to \RR$, $\hat{u}_a:[0,T]\times[0,1]\times \RR \times \mathcal{P}([0,1]\times \mathbb{R})\times \RR \to \RR$, $\hat{\sigma}: [0,T]\times [0,1]\times \RR \to \RR$, and $\hat{\mu} :[0,T] \to \mathcal{P}([0,1]\times \mathbb{R})$ such that
\begin{align*}
    \mu \circ (U,X_t)^{-1} &= \hat{\mu}_t ,\\
    {b}_t\left(U,X,\mu,\hat{\alpha}_t(U,X,Y_t,\mu,Z_t,\chi_t)\right) &=\hat{b}(t,U,X_t,Y_t,\hat{\mu}_t,Z_t,\chi_t),\\
   {c}_t\left(U,X,\mu,\hat{\alpha}_t(U,X,Y_t,\mu,Z_t,\chi_t)\right)&=\hat{c}(t,U,X_t,Y_t,\hat{\mu}_t,Z_t,\chi_t),\\
    {u}_a(t,U,X,\mu,\chi_t) &=\hat{u}_a(t,U,X_t,\hat{\mu}_t,\chi_t), \\
    \sigma_t(U,X) &= \hat{\sigma}(t,U,X_t).
\end{align*}
We introduce the principal's dynamic value function $v_P:[0,T]\times \mathcal{P}^2(\RR^3
 )\to \RR$ defined by
\begin{align*}
    v_P(t,\nu) &:= \sup_{Z,\chi}\EE^{\PP}
    \left[U_p\left(X_T-U_a^{-1}(Y_T)-\int_t^T{\chi}_sds\right)\right],
    \end{align*}
For all $0 \leq t\leq s\leq T$, the triple $(X_s,Y_s,\mu_s^{t,\nu})$ solves the McKean--Vlasov equation below.
\begin{align}
		X_s &=X_t +\int_t^s \hat{b}_r\Big(U,X_r,Y_r,\mu^{t,\nu}_r,Z_r,\chi_r\Big)dr + \int_t^s\hat{\sigma}_r(U,X_r)dB_r,\quad \mathbb{P}-a.s., \nonumber\\
			Y_s &= Y_t +\int_t^s\Big(\hat{b}_r(U,X_r,Y_r,\mu^{t,\nu}_r,Z_r,\chi_r)Z_r-H_r(U,X_r,Y_r,\mu^{t,\nu}_r,Z_r,\chi_r)\Big)dr+\int_t^sZ_r\hat{\sigma}_r(U,X_r)dB_r, \nonumber \\
              \mu_s^{t,\nu}&:= \PP \circ (U,X_s,Y_s)^{-1},\nonumber \\
            \nu &:= \PP \circ (U,X_t,Y_t)^{-1}. \label{mckean.1}
\end{align}

In addition, for all $(z,\chi,\nu) \in \RR^2\times \mathcal{P}^2\left( \RR^3\right)$, $\phi \in {C}^2\left(\mathcal{P}^2(\RR^3) \right)$, and $(u,x,y) \in [0,1]\times \RR^2$ define
\begin{align*}
\mathcal{L}^{e,z}\phi(\nu)(u,x,y)&:= \partial_{\nu}\phi(\nu)(u,x,y)\cdot B(t,u,x,y,\nu,z,e) \\
&\quad + \frac{1}{2}\text{Tr}\left(\partial_{(u,x,y)}\partial_{\nu}\phi(\nu)(u,x,y) SS^{\top}(t,u,x,z)\right)-e,
\end{align*}
where
\begin{align*}
    B(t,u,x,y,\nu,z,e) &:=\left(0,\hat{b}\left(t,u,x,y,\nu,z,e\right),\hat{b}\left(t,u,x,y,\nu,z,e\right)z-H(t,u,x,y,\nu,z,e)\right) ,\\
    S(t,u,x,z) &:= \begin{pmatrix}
0 & 0 & 0  \\
0 & {\hat{\sigma}}(t,u,x) & 0 \\
0 & z{\hat{\sigma}}(t,u,x)& 0
\end{pmatrix}.
\end{align*}
Following Pham and Wei \cite{bellmanphamwei}, we introduce the Hamilton-Jacobi-Bellman (HJB) equation associated with the mean-field control problem above.

\begin{align}
   &\partial_{t} v(t,\nu) +\int_{[0,1]\times \RR^2}\left(\sup_{(e,z)\in \mathbb{R}^2}\mathcal{L}^{e,z}v(t,\nu)(u,x,y)\right)\nu(du,dx,dy)= 0,\label{hjb.eqn}\\
   &v(T,\nu) = \int_{[0,1]\times \RR^2}U_p\left(x-U_a^{-1}(y)\right)\nu(du,dx,dy). \nonumber
\end{align}
We conclude this section with a verification lemma adapted from Theorem 4.2 in Élie et al. \cite{manymanyagents}. The result allows us to construct optimal contracts from classical solutions of the HJB equation \eqref{hjb.eqn}.
\begin{lemma}[Verification]\label{thm.verification}
	Let Assumptions \ref{ass1}, \ref{ass3}, and
	\ref{ass.max.hamiltonian} hold. Suppose that
	\(w\in C^{1,2}([0,T]\times\mathcal P^2(\mathbb R^3))\) solves
	\eqref{hjb.eqn} and satisfies, for some \(C>0\),
	\[
	|w(t,\mu)|\le C(1+\|\mu\|_2^2),\qquad
	|\partial_\mu w(t,\mu)(u,x,y)|
	\le C(1+|(u,x,y)|+\|\mu\|_2).
	\]
	For each measurable function \(\eta: [0,1]\times \RR \to \RR\), set
	$
	\mu_0^\eta:=\lambda(du,dx)\delta_{\eta(u,x)}(dy)$,
	and assume that
	$$
	\eta^*\in\argmax_{\eta:\,\eta(u,x)\ge R_a(u,x)} w(0,\mu_0^\eta).
	$$
	Assume further that the Hamiltonian in \eqref{hjb.eqn} admits a measurable
	maximizer \((\chi^*,z^*)\), and that the closed-loop McKean--Vlasov system
	\[
	Y_0^*=\eta^*(U,X_0),\qquad
	\mu_t^*=\mathcal L(U,X_t^*,Y_t^*),\qquad
	Z_t^*=z^*(t,U,X_t^*,Y_t^*,\mu_t^*), \qquad 	\chi_t^*=\chi^*(t,U,X_t^*,Y_t^*,\mu_t^*),
	\]
	admits a solution \((U,X^*,Y^*,\mu^*)\). Define
	\[
	\alpha_t^*
	=
	\hat\alpha(t,U,X_t^*,Y_t^*,\mu_t^{*},Z_t^*,\chi^*_t),
	\qquad
	\xi^*=U_a^{-1}(Y_T^*),
	\]
	and assume the admissibility condition
	\[
	Z^*\in\mathbb H_{\exp}([0,T]),
	\qquad
	\xi^*\in \mathbb{L}_{\exp}(\mathcal F_T).
	\]
	Then $(\chi^*,\xi^*)\in \Sigma_a$ is optimal, and
	\[
	(\alpha^*,\mu^{*})\in\mathcal M^*(\chi^*,\xi^*),
	\qquad
	V_p=w(0,\mu_0^{\eta^*}).
	\]
\end{lemma}
\begin{remark}\label{rem.hjb.alternative}
The HJB equation is used here as a verification device, rather than as the only possible construction method. In examples where the associated mean-field control problem can be solved by Pontryagin-type arguments (see \cite{Carmonabook}), the resulting optimizer can be inserted into the parametrization \eqref{mckean.1} to construct the contract.
\end{remark}
    \section{Main results}\label{section.last}
 We introduce a contracting model with a principal and multiple non-cooperative agents who manage heterogeneously interdependent projects. The model captures how incentives should vary across agents with different levels of influence. We present two versions of the model: the $N$-agent version and the heterogeneous mean-field version, using the general method developed in Section~\ref{section.technicalsetup}. In each case, we compute the optimal contracts and the agents' optimal efforts; see Theorem \ref{opt.contract.finite.example1} and Theorem \ref{opt.contract.interaction.example.1}. We also establish that the $N$-agent problem converges to the heterogeneous mean-field problem as $N\rightarrow \infty$ in the following sense:
 \begin{itemize}
    \item[$i)$] The agents' optimal efforts converge as $N \rightarrow \infty$; see Theorem \ref{lemmaQ}.
    \item[$ii)$] The optimal contracts and the principal's value converge as $N\rightarrow \infty$; see Theorem \ref{conv.Value}.
    \item[$iii)$] The optimal contract for a continuum of agents generates a near-optimal contract for the $N$-agent problem when $N$ is large enough. By near-optimal, we mean that the contract satisfies the agents' reservation constraints, induces a unique Nash equilibrium effort in the $N$-agent problem, and is $\mathcal{O}(1/N)$-optimal from the principal's viewpoint; see Theorem \ref{nearoptimal}.
\end{itemize}
We further show that optimal contracts incentivize more influential agents to exert greater effort (Theorem \ref{Lemma.qualitative}). We also prove that optimal compensations are stable under perturbations of the agents' interactions (Theorem \ref{lemma.Q.stability}, Theorem \ref{theorem.contracts.stability}).   

	 We model heterogeneity through an interaction function $G:[0,1]^2\to \RR$, representing the agents' interdependence, and $\lambda \in \mathcal{P}\left([0,1]\times \RR\right)$, denoting the joint distribution of type and initial output. We denote by $\hat{\lambda}(u,dx) \in \mathcal{P}( \RR)$ the disintegration of $X$ given $U=u$, which gives the distribution of the initial output of type $u\in [0,1]$. In this section, we make the following assumptions on $G$ and $\hat{\lambda}$.
	\begin{ass}\label{ass.G}
		The mappings $G:[0,1]^2\to \RR$ and $\hat{\lambda}:[0,1]\to \mathcal{P}(\RR)$ satisfy the following properties:
		\begin{itemize}
        \item[(a)] $\hat{\lambda}$ is measurable.
			\item[(b)]  $\sup_{u\in [0,1]} \int_{\RR}x^2\hat{\lambda}(u,dx)<\infty$.
            \item[(c)] There exist $C>0$, $m \in \mathbb{N}$, and a collection of intervals $(I_i)_{i=1}^m$ such that         $\cup_{i=1}^mI_i=[0,1],$ and
            \begin{align*}
            \big|G(u_1,v_1)-G(u_2,v_2)\big|&\leq C\left(|u_1-u_2|+|v_1-v_2| \right), \quad  (u_1,v_1), (u_2,v_2) \in I_i\times I_j, \quad i,j\in [m], \\
\mathcal{W}_2\left(\hat{\lambda}({u_1},dx),\hat{\lambda}({u_2},dx)\right)&\leq C|u_1-u_2|, \quad u_1, u_2 \in I_i, \quad i \in [m],
            \end{align*}
where $\mathcal{W}_2$ denotes the Wasserstein--$2$ distance on $\mathcal{P}^2(\RR)$.
		\end{itemize}
	\end{ass}
In addition, for all $G$ satisfying Assumption \ref{ass.G}, we define $\|G\|_{\infty}:= \sup_{u,v\in[0,1]}|G(u,v)|$. We denote by $\tilde{R}_a \in \mathcal{F}_0$ the reservation utility of the representative agent, which satisfies the following property.
\begin{ass}\label{ass.reservation}
There exists a Lipschitz continuous function ${R}_a : [0,1]\to \RR$ such that $\tilde{R}_a = {R}_a(U)$, $\lambda-a.s.$
\end{ass}

\subsection{Contracting $N$ heterogeneous agents}\label{example1.finite.agents}
Let $N \in \mathbb{N}$ be the number of agents, and let $[N]:=\{1,\ldots,N\}$. We fix $\Omega^N:= C([0,T],\RR^N)$ endowed with the Wiener measure $\WW$, and denote by $\bar{B}^N:=(B^{1,N},\ldots,B^{N,N})$ the canonical Brownian process. Given deterministic initial project values $\bar{x}_0^N:=(x_0^{1,N},\ldots,x_0^{N,N})\in\RR^N$, we set $X^{i,N}_t=x_0^{i,N}+B^{i,N}_t$ under $\WW$. We define $G^N_{i,j}:= G(\frac{i}{N},\frac{j}{N})$, $i,j \in [N]$, where $G$ satisfies Assumption \ref{ass.G}, and introduce $\alpha:= (\alpha^1,\ldots,\alpha^N)$ for the agents' effort processes. Under effort $\alpha$, the agents control the output process $\bar{X}^N :=(X^{1,N},\ldots,X^{N,N})$ as follows:
\begin{align}\label{state.finite.1}
		X^{i,N}_t =x_0^{i,N}+\int_0^t\left(\frac{1}{N}\sum_{j= 1}^N G^N_{i,j}X^{j,N}_s+  \alpha^i_s \right)ds +W^{\alpha,i}_t, \quad \PP^\alpha-a.s.,\quad i \in [N],\end{align}
	where $W^\alpha:=(W^{\alpha,i})_{i=1}^N$ is an $N$-dimensional standard $\PP^{\alpha}$-Brownian motion. Here $\PP^\alpha$ is defined by the following Radon-Nikodym derivative:
    \begin{align*}
    &\frac{d\PP^{\alpha}}{d\WW}:=\exp\left(-\frac{1}{2}\sum_{i=1}^N\int_0^T\left(\frac{1}{N}\sum_{j= 1}^{N} G^N_{i,j}X^{j,N}_s+  \alpha^i_s \right)^2ds+\sum_{i=1}^N\int_0^T \left(\frac{1}{N}\sum_{j= 1}^N G^N_{i,j}X^{j,N}_s+  \alpha^i_s \right)dB^{i,N}_s \right),
    \end{align*}
    We denote by $\mathcal{A}$ the set of $\RR^N$-valued, $\FF^{\bar{B}^N}$-adapted processes for which there exists $\epsilon>0$ such that $$\EE^{\WW}\left[\left(\frac{d\PP^{\alpha}}{d\WW}\right)^{1+\epsilon} \right] <\infty,$$
    and, for all $p\geq 0$ $$\EE^{\WW}\left[\exp\left(p\int_0^T|\alpha_s|^2ds\right) \right]<\infty.$$
\begin{remark}\label{rem.random.initial.finite}
We state the finite-agent model with deterministic initial conditions to simplify notation. The same arguments apply to independent square-integrable initial conditions with laws $\hat\lambda(i/N,dx)$, which is the formulation used in the convergence estimates.
\end{remark}
    Next, we introduce the set of contracts $\Sigma^N$ defined as the set of $\mathcal{F}_T^{\bar{B}^N}$-measurable random vectors
    $\bar{\xi}^N:= (\xi^{1,N},\ldots,\xi^{N,N})$  satisfying
    \begin{align}
        &\sup_{i \in [N]}\EE^{\WW}\left[\exp\left(p|\xi^{i,N}|\right)\right]<\infty, \quad \forall p\geq0.
    \end{align}
    Then, given the contract vector $\bar{\xi}^N\in \Sigma^N$, the agents maximize the expected value of their contracts minus their cost of effort.
   \begin{equation*}	J^{i}_a(\alpha^i,\xi^{i,N}) := \EE^{\PP^{\alpha}} \left[\xi^{i,N} -\frac{1}{2}\int_0^T(\alpha^i_s)^2ds \right], \quad i \in [N].
\end{equation*}
We assume that agents are non-cooperative, so their efforts are chosen in Nash equilibrium. We denote by $\mathcal{A}^*_N(\bar{\xi}^N)$ the set of agents' Nash equilibria given the contracts $\bar{\xi}^N$. We assign to every agent $i \in [N]$ the reservation utility $R_a(i/N)$, where $R_a:[0,1]\to \RR$ satisfies Assumption \ref{ass.reservation}. We also denote by $\Sigma_a^N$ the set of contracts that satisfy the reservation utility constraint, i.e., $\bar{\xi}^N \in \Sigma_a^N$ if and only if $\mathcal{A}_N^*(\bar{\xi}^N)\neq\emptyset$ and, for every $\alpha:=(\alpha^1,\ldots,\alpha^N)\in \mathcal{A}_N^*(\bar{\xi}^N)$,
\begin{align*}
    J_a^i(\alpha^i,\xi^{i,N}) \geq R_a(i/N), \quad i \in [N].
\end{align*}

	The principal, in turn, maximizes the expected average terminal project value net of compensation.
	\begin{equation}
		J^N_p(\bar{\xi}^N) := \sup_{\alpha \in \mathcal{A}^*_N(\bar{\xi}^N)}\EE^{\PP^\alpha}\left[\frac{1}{N}\sum_{i=1}^NX^{i,N}_T-\frac{1}{N}\sum_{i=1}^N\xi^{i,N} \right], \quad V^N_p:=\sup_{\bar{\xi}^N\in \Sigma_a^N}J^N_p(\bar{\xi}^N)
	\end{equation}
	We apply Élie and Possamaï {\cite{competitive}} to reduce the principal's problem to a standard stochastic control problem. For every agent $i \in [N]$, we introduce the maps $h^i: [0,T]\times \RR^N \times \RR^{N\times N}\times \RR^N \to \RR$ and $H^i: [0,T]\times\RR^N\times \RR^{N\times N}  \to \RR$ defined as follows:
	\begin{align*}
		h^i(t,{x},{z},{a}) &:= \sum_{j=1}^Nz_{i,j}\left(\frac{1}{N}\sum_{k= 1}^NG^N_{j,k}x_k+a_j\right)-\frac{1}{2}(a_i)^2, \\
		H^i(t,{x},{z}) &:= h^i(t,x,z,a^*)= \frac{1}{2} z_{i,i}^2+\sum_{j\neq i}^Nz_{i,j}z_{j,j}+\sum_{j=1}^Nz_{i,j}\frac{1}{N}\sum_{k=1}^NG^N_{j,k}x_k,
	\end{align*}
	where $a^*:=(z^{11},\ldots,z^{NN})$ is the unique solution of the fixed-point equation
    \begin{equation*}(a^1,\ldots,a^N)=\left(\argmax h^1(t,x,z,a^{-1},\cdot),\ldots,\argmax h^N(t,x,z,a^{-N},\cdot)\right).
    \end{equation*}
    Indeed, the first-order conditions give $a^i=z_{i,i}$ for every $i$, so the fixed point is well defined and unique.
 Let $Z^N$ be an $\FF$-adapted,  $\RR^{N\times N}$-valued process. Given $(Y_0,Z)$, introduce the controlled process $Y^{Y_0,Z}:=(Y^{1,Y_0,Z},\ldots,Y^{N,Y_0,Z})$ defined by
	\begin{align*}
		dY^{i,Y_0,Z}_t &= -H^i_t(\bar{X}^N_t,{Z}_t)dt +{Z}^i_t d\bar{X}^N_t \\
		&= -\left(\frac{1}{2}(Z^{i,i}_t)^2+\sum_{j\neq i}Z^{i,j}_tZ^{j,j}_t+\frac{1}{N}\sum_{j= 1}^NZ_{t}^{i,j}\sum_{k=1}^NG^N_{j,k}X^{k,N}_t \right)dt+{Z}^i_t\cdot d\bar{X}^N_t, \quad i \in [N] \\
		&=\frac{1}{2}(Z^{i,i}_t)^2dt+{Z}^i_t\cdot d\bar{W}^{\bar{Z}}_t,\quad i \in [N],
	\end{align*}
and $Y_0\in \RR^N$. Here $\bar W^{\bar Z}$ denotes the Brownian motion under the probability measure induced by the equilibrium effort $\bar Z=(Z^{1,1},\ldots,Z^{N,N})$. We write ${Z}^i:=(Z^{i,1},\ldots,Z^{i,N})$ for all $i \in [N]$.

The first-order condition for agent $i$ gives $a^{i,*}=Z^{i,i}$, so the induced effort profile is pinned down only by the diagonal entries of $Z$. Under $\PP^{\bar Z}$,
\[
    dY^{i,Y_0,Z}_t
    =
    \frac12 (Z^{i,i}_t)^2dt
    +
    Z^i_t\cdot d\bar W^{\bar Z}_t .
\]
The off-diagonal entries $Z^{i,j}$, $j\neq i$, change only the risk of the terminal output without affecting its expectation or the optimal effort. Since the principal's payoff is linear, this additional risk has no effect. Consequently, for any admissible contract represented by $(Y_0,Z)$, replacing $Z$ by its diagonal part yields the same equilibrium effort and the same expected total payment. Therefore, the principal's control processes can be reduced without loss of optimality to the diagonal process $\bar Z^N=(Z^{1,1},\ldots,Z^{N,N})$.

Next, we fix the probability space $(\Omega,\mathbb{F},\mathbb{P})$, on which we define an $N$-dimensional standard $\PP$-Brownian motion $\tilde{W}:=(\tilde{W}^1,\ldots,\tilde{W}^N)$. Let $\bar{Z}:= (Z^{11},\ldots,Z^{NN})$ be an $\FF$-adapted, square-integrable process defined on $(\Omega,\mathbb{F},\PP)$. The principal solves the following linear-quadratic stochastic control problem:
	\begin{equation*}
		V^N_p:= \sup_{\bar{Z} }\EE\left[\frac{1}{N}\sum_{j=1}^N  X^{j,N}_T- \frac{1}{2 N} \sum_{j=1}^N \int_0^T(Z^{j,j}_s)^2 ds\right]-\frac{1}{N}\sum_{j=1}^N R_a(j/N),
	\end{equation*}
	where
	\begin{align*}
		dX^{i,N}_t &= \left( \frac{1}{N}\sum_{j=  1}^N G^N_{i,j}X^{j,N}_t + Z^{i,i}_t\right)dt +  d\tilde{W}^{i}_t, \quad i \in [N],  \\
		X_0^{i,N} &=x_0^{i,N}.
	\end{align*}
	We introduce the dynamic value function $v : [0,T]\times \RR^N \to \RR$ defined by
    \begin{align}\label{control.finite.agents}
        v(t,{x}) := \sup_{\bar{Z}}\mathbb{E}\left[\frac{1}{N}\sum_{j=1}^NX^{j,N}_T-\frac{1}{2N}\sum_{j=1}^N\int_t^T(Z^{j,j}_s)^2ds \bigg|\bar{X}_t^N={x}\right].
    \end{align}
The Hamilton-Jacobi-Bellman (HJB) equation associated with the control problem \eqref{control.finite.agents} is
	\begin{align*}
		&v_t + \sup_{\bar{z}\in \RR^N} \bigg\{ \sum_{i=1}^N\bigg(\Big(\frac{1}{N}\sum_{j=1  }^NG^N_{i,j}x_j+ z^i\Big)v_{x_i}+\frac{1}{2}v_{x_i x_i} -\frac{1}{2N}(z^i)^2\bigg)\bigg\} = 0, \\
		&v(T,{x}) = \frac{1}{N}\sum_{j=1}^N x_j.
	\end{align*}
	This HJB equation leads to the following parabolic PDE:
	\begin{align*}
		&v_t + \frac{1}{N}\sum_{i=1}^N\left(\sum_{j= 1 }^N G^N_{i,j}x_j\right)v_{x_i}+\frac{N}{2}\sum_{i=1}^Nv_{x_i}^2+\frac{1}{2}\sum_{i=1}^N  v_{x_i x_i} = 0, \\
		&v(T,{x}) = \frac{1}{N}\sum_{j=1}^N  x_j.
	\end{align*}
      It is straightforward to verify that the value function satisfies
	\begin{equation*}
		v(t,{x}) = \frac{1}{N}\sum_{i=1}^NQ^{i,N}(t)x_i+\frac{1}{2N}\sum_{i=1}^N \int_t^T\left(Q^{i,N}(s)\right)^2ds,
	\end{equation*}
	where $\bar{Q}^N:=(Q^{1,N},\ldots,Q^{N,N})\in C([0,T],\RR^N)$ is the unique solution of the following linear ODE system
	\begin{align}\label{ODE.system}
		\dot{Q}^{i,N}(t)  &= -\frac{1}{N}\sum_{j= 1}^NG^N_{j,i}Q^{j,N}(t), \quad i \in [N],\\
		Q^{i,N}(T) &=1, \quad i \in [N]. \nonumber
	\end{align}
    This ODE system characterizes the solution of the contracting problem. The next theorem summarizes the result.
    \begin{theorem}\label{opt.contract.finite.example1}
         Let $\bar{Q}^N\in C([0,T],\RR^N)$ be the unique solution of the linear system \eqref{ODE.system}. Then,
         \begin{itemize}
         \item[a)] $V^N_p = \frac{1}{N}\sum_{i=1}^NQ^{i,N}(0){x}_0^{i,N}+\frac{1}{2N}\sum_{i=1}^N\int_0^T \left(Q^{i,N}(t)\right)^2dt-\frac{1}{N}\sum_{i=1}^N R_a(i/N)$.
         \item[b)] The contracts $\bar{\xi}^{N,*}:=(\xi^{1,N,*},\ldots,\xi^{N,N,*})$ defined as
         \begin{align*}
		\xi^{i,N,*}&:= R_a(i/N)- \int_0^T\left(\frac{1}{2}(Q^{i,N}(s))^2+Q^{i,N}(s)\frac{1}{N}\sum_{j=1}^NG_{i,j}^NX_s^{j,N}\right)ds\\
		            &\quad +\int_0^T{Q}^{i,N}(s) d{X}^{i,N}_s, \quad i \in [N],
	\end{align*}
    are optimal.
    \item[(c)] Under $\PP^{\bar{Q}^N}$, the optimal contracts $\bar{\xi}^{N,*}$ form an independent
    Gaussian vector with
    \begin{align*}
        \EE^{\PP^{\bar{Q}^N}}[\xi^{i,N,*}]
        &=R_a(i/N) + \frac{1}{2}\int_0^T \left(Q^{i,N}(s)\right)^2ds,\\
        \operatorname{Var}^{\PP^{\bar{Q}^N}}[\xi^{i,N,*}]
        &= \int_0^T \left(Q^{i,N}(s)\right)^2ds, \quad i \in [N].
    \end{align*}
    \item[d)] $V_a^i({\xi}^{i,N,*})=R_a(i/N),\quad$ $i \in [N]$.
    \item[e)] $\alpha_t^{i,N,*}=Q^{i,N}(t)$, $\quad (i,t) \in [N]\times [0,T]$.
    \end{itemize}
    \end{theorem}
	\subsection{Contracting a continuum of heterogeneous agents}\label{continuum.agents.example}
We introduce the heterogeneous mean-field version of the principal-agent model studied previously. The technical setup and notation are described in Section \ref{section.technicalsetup}. Let $\mu \in C\left([0,T],\mathcal{P}^2_{\text{Unif}}\left( [0,1]\times \RR  \right)\right)$ be a flow of measures such that $\mu_0 = \lambda$, where $\lambda$ admits a disintegration $\hat{\lambda}$ satisfying Assumption \ref{ass.G}. We describe the dynamics of the representative agent's project as follows:
	\begin{align*}
		&dX_t =\left(\int_{[0,1]\times\RR}G(U,v)x \mu_t(dv,dx)+\alpha_t \right)dt + dW^{\alpha,\mu}_t, \quad \PP^{\alpha,\mu}-a.s.,	\end{align*}
	where
    \begin{align}
        \frac{d\PP^{\alpha,\mu}}{d\PP_0} := \mathcal{E}\left(\int_0^\cdot\left(\int_{[0,1]\times\RR}G(U,v)x \mu_t(dv,dx)+\alpha_t \right) dW_t \right)_T,
    \end{align}
   and $(X_0,U) \sim \lambda$, which satisfies Assumption \ref{ass.G}. We use the same interaction function $G$ as in the previous subsection.

    Given a contract $\xi \in \Sigma_a$, the representative agent maximizes the following linear-quadratic objective
	\begin{equation*}
		J_a(\alpha,\mu,\xi) := \EE^{\PP^{\alpha,\mu}} \left[\xi-\frac{1}{2}\int_0^T \alpha_s^2ds\right], \quad
        V_a(\mu,\xi) = \sup_{\alpha \in \mathcal{A}}J_a(\alpha,\mu,\xi).
	\end{equation*}
    In addition, we assume that the reservation utility of the representative agent is given by the random variable $\tilde{R}_a = R_a(U)$, where $R_a : [0,1]\to \RR$ is the function used in the previous subsection and satisfies Assumption \ref{ass.reservation}. The principal, in turn, optimizes the following expression:
	\begin{equation}
		J_p(\xi) := \sup_{(\alpha,\mu) \in \mathcal{M}^*(\xi)}\EE^{\PP^{\alpha,\mu}} \left[X_T-\xi\right].
	\end{equation}
Let $(Y_0,Z^1)$ be given. We denote by $Y^{Y_0,Z^1}$ the following controlled process
\begin{align*}
		Y_t &= Y_0+\frac{1}{2}\int_0^t (Z^1_s)^2ds +\int_0^t Z^1_s dW^{Z,\mu}_s, \quad \PP^{\alpha^*(Y_0,Z^1),\mu}-a.s.,\quad t\in [0,T],
	\end{align*}
    where $(X,\mu)$ solve the McKean--Vlasov SDE
 \begin{align*}
				dX_t &=\left(\int_{[0,1]\times\RR}G(U,v)x \mu_t(dv,dx)+Z^1_t \right)dt +dW^{Z,\mu}_t, \quad \PP^{Z,\mu}-a.s.,	\\
                \mu_t &= \PP^{Z,\mu}\circ (U,X_t)^{-1}.
	\end{align*}
Using Corollary \ref{theorem.mean.field.control}, we can write the principal's problem as the following mean-field control problem:
	\begin{align}\label{meanfieldcontrol.example}
		\hat{J}_p(Y_0,Z) :=J_p(Y^{Y_0,Z}_T)= \EE^{\PP^{Z,\mu}}\left[X_T-\frac{1}{2}\int_0^TZ_t^2dt \right]-\EE^{\PP^{Z,\mu}}[Y_0].
	\end{align}
   First, note that $Y^*_0 = R_a(U)$. We then find $Z^*$ by verification. We fix the probability space $(\Omega,\mathbb{F},\mathbb{P})$, on which we construct a $\PP$-Brownian motion $W$ and a random vector $(U,X_0)\sim \lambda$, independent of $W$. We write the principal's dynamic value function as
   \begin{align}\label{value.state.measures}
		{v}(t,\mu) &:= \sup_{Z}\EE\left[X_T-\frac{1}{2}\int_t^T(Z_s)^2ds \right],
    \end{align}
    For all $0 \leq t \leq s\leq T$, $(X_s,\mu_s^{t,\mu})$ denotes the solution of the McKean--Vlasov SDE
    \begin{align}\label{controlledMcKeanExample}
        X_s &= X_t+\int_t^s\left( \int_{[0,1]\times \RR}G(U,u)x\mu^{t,\mu}_{r}(du,dx)+Z_r\right)dr+ W_s-W_t,\\
        \mu_s^{t,\mu} &= \PP \circ  (U,X_s)^{-1}, \nonumber\\
        (U,X_t) &\sim \mu.  \nonumber
    \end{align}
    By Remark $4.22$ in Carmona and Delarue \cite{Carmonabook}, equation \eqref{controlledMcKeanExample} admits a unique solution for all square-integrable, $\FF$-adapted processes $Z$.

Finally, we use Lemma \ref{thm.verification}. First, we introduce the Hamilton-Jacobi-Bellman (HJB) equation associated with the control problem \eqref{value.state.measures}.
    \begin{align}
        \partial_t v(t,\mu) &+ \int_{[0,1]\times \RR}\sup_{z\in \RR}\left\{\Big( \EE^{\mu}[G(u,U)X]+z\Big)\partial_{\mu_2}v(t,\mu)(u,x)-\frac{1}{2}z^2 \right\}\mu(du,dx) \nonumber \\
        &+\frac{1}{2}\EE^\mu\left[\left(\partial_{x}\partial_{\mu_2}v(t,\mu)(U,X)\right)\right]= 0,\label{hjb.model}
    \end{align}
    with terminal condition  $v(T,\mu) = \EE^{\mu}[X]$.
      In this equation, $(U,X) \sim \mu$, and $\partial_{\mu_2}v(t,\mu)$ denotes the second component of the lifted derivative on $\mathcal{P}^2(\RR^2)$. The HJB equation \eqref{hjb.model} can be rewritten as follows:
    \begin{align}\label{hjb.model2}
\partial_tv(t,\mu) &+ \frac{1}{2}\EE^{\mu}\left[\big({\partial_{\mu_2}}v(t,\mu)(U,X)\big)^2\right]+\int_{[0,1]\times\RR}\EE^{\mu}\left[G(u,U)X\right]\left({\partial _{\mu_2}}v(t,\mu)(u,x)\right)\mu(du,dx)\\
		&+\frac{1}{2}\EE^{\mu}\left[\left({\partial_x}{\partial_{\mu_2}}v(t,\mu)\right)(U,X)\right]= 0, \nonumber
	\end{align}
    and $v(T,\mu) = \EE^{\mu}[X]$.

   We solve this PDE using the ansatz
    $$w(t,\mu) := \EE^{\mu}\left[Q\left(t,U\right){X}\right]+P(t),$$ where $({U},{X})\sim \mu$, and  $Q: [0,T]\times [0,1]\to \RR$, $P:[0,T]\to \RR$ are functions to be determined. We start by computing the following quantities:
	\begin{align*}
w_t(t,\mu)&=\EE^{\mu}\left[\dot{Q}(t,U)X\right]+\dot{P}(t),\\
		\int_0^1\EE^{\mu}\left[G(u,U)X\right] \partial_{\mu_2}w(t,\mu)(u)du&= \int_0^1\EE^{\mu}\left[G(u,U)Q(t,u)X\right]du, \\
\EE^{\mu}\left[\big({\partial_{\mu_2}}w(t,\mu)(U,X)\big)^2\right]&=\int_0^1Q(t,u)^2du, \\
\frac{1}{2}\EE^{\mu}\left[\left({\partial_x}{\partial_{\mu_2}}w(t,\mu)\right)(U,X)\right]&=0.
	\end{align*}
	Substituting these expressions into the HJB equation \eqref{hjb.model2}, we obtain
	\begin{align*}
		&\int_{ [0,1]\times \RR}x\Big(\dot{Q}(t,u)+\int_0^1G(v,u)Q(t,v)dv\Big)\mu(du,dx)= 0, \quad t\in [0,T], \\  &{P}(t) = \frac{1}{2}\int_t^T\int_0^1Q(s,u)^2duds,\quad t\in [0,T].
	\end{align*}
	This equation is satisfied when $Q$ solves the following infinite-dimensional ODE system:
	\begin{align}\label{infinite.system}
		&\dot{Q}(t,u) = -\int_0^1G(v,u)Q(t,v)dv,\quad u \in [0,1],\\
		&Q(T,u) = \mathbbm{1}_{[0,1]}(u).  \nonumber
	\end{align}
	Under Assumption \ref{ass.G}, this system admits a unique solution $Q \in C\left([0,T],L^2([0,1])\right)$. Moreover, for all $t\in [0,T]$, $Q(t,\cdot)$ is piecewise Lipschitz continuous on each interval $I_i, i \in [m]$, defined in Assumption \ref{ass.G}. Thus,
    \begin{align*}
        v(t,\mu) = \EE^{\mu}[Q(t,U)X]+P(t),
    \end{align*}
    where $Q$ solves \eqref{infinite.system}, and $P(t)=\frac{1}{2}\int_t^T\int_0^1Q(s,u)^2duds$.
	It remains to show that $\big(Q(\cdot,U),\PP\circ (U,X)^{-1}\big)$ is a heterogeneous mean-field equilibrium for the representative-agent problem. Applying Remark $4.22$ in Carmona and Delarue \cite{Carmonabook} (Volume I), there exists a unique solution of the following McKean--Vlasov SDE:
    \begin{align}\label{mckean.vlasov.system}
		dX_t &= \left( Q(t,U)+ \int_{[0,1]\times\RR}G(U,v)x\mu_t^*(dv,dx)\right)dt +dW_t, \\
        \mu^*_t(dv,dx)&=\PP\circ \left(U,X_t\right)^{-1}, \nonumber\\
		(U,X_0) &\sim \lambda. \nonumber
	\end{align}
    Using Lemma \ref{thm.verification}, we obtain the following optimal contract:
    \begin{align*}\label{eqn.McKeanVlasov}
        \xi^* &:= R_a(U)-\int_0^T \left(\frac{1}{2}Q(t,U)^2+Q(t,U)\int_{[0,1]\times \RR}G(U,v)x\mu^*_t(dv,dx) \right)dt +\int_0^TQ(s,U)dX_s.
    \end{align*}
    Moreover, $(Q,\mu^*)$ defines a heterogeneous mean-field equilibrium given the contract $\xi^*$, where $Q$ solves \eqref{infinite.system} and $\mu^*$ is defined in terms of the unique solution of the McKean--Vlasov equation \eqref{mckean.vlasov.system}.

    Thus, the solution of the heterogeneous principal-agent problem is characterized by the solution $Q$ of the infinite-dimensional system \eqref{infinite.system}. The following theorem summarizes the result.
    \begin{theorem}\label{opt.contract.interaction.example.1}
         Suppose Assumption \ref{ass.G} holds, and let
        $Q \in C([0,T],L^2([0,1]))$ be the unique solution of the linear system \eqref{infinite.system}. Then,
        \begin{itemize}
        \item[a)] $V_p = \int_{[0,1]\times \RR}Q(0,u)x\lambda(du,dx)+\frac{1}{2}\int_0^T\int_{0}^1Q(t,u)^2dudt -\int_0^1R_a(u)du$.
        \item[b)] The following contract is optimal:
        \[
        \xi^* = R_a(U)-\int_0^T \left(\frac{1}{2}Q(t,U)^2+Q(t,U)\int_{[0,1]\times \RR}G(U,v)x\mu^*_t(dv,dx) \right)dt +\int_0^TQ(s,U)dX_s.
        \]
        Here $\mu^*$ is determined uniquely by the McKean--Vlasov equation \eqref{mckean.vlasov.system}.
           \item[c)] $\PP^{Q(\cdot,U)}\left( \xi^*\in \cdot |U=u\right)$ is Gaussian for all $u \in [0,1]$, with
        \begin{align*}
            \EE^{\PP^{Q(\cdot,U),\mu^*}}\left[\xi^*\mid U=u\right]
            &=R_a(u)+\frac{1}{2}\int_0^T Q(t,u)^2dt,\\
            \operatorname{Var}^{\PP^{Q(\cdot,U),\mu^*}}\left[\xi^*\mid U=u\right]
            &=\int_0^T Q(t,u)^2dt.
        \end{align*}
        \item[d)] $V^a_0(\xi^*) = R_a(U), \lambda -a.s.$
        \item[e)] $\alpha^*(t,u) = Q(t,u)$, $(t,u) \in [0,T]\times [0,1]$.
        \end{itemize}

    \end{theorem}
\subsubsection{Structure of the optimal contracts}
The representation in Theorem \ref{opt.contract.interaction.example.1} shows that the optimal contract separates into a fixed component and a performance-based component:
\begin{align*}
\xi^*
=
&\underbrace{R_a(U)-\int_0^T
\left[
\frac{1}{2}Q(t,U)^2
+
Q(t,U)
\int_{[0,1]\times \RR}G(U,v)x\mu_t^*(dv,dx)
\right]dt}_{\text{fixed component}}\\
&\quad+
\underbrace{\int_0^T Q(t,U)dX_t}_{\text{performance-based component}}.
\end{align*}
The fixed component of the contract is set so that the agent's participation constraint binds. It guarantees the required reservation utility while also accounting for the deterministic part of the output process and the cost of the effort induced by the contract. The performance-based component is linear in the realized output path. Its coefficient is $Q(t,U)$, which also corresponds to the equilibrium effort chosen by an agent of type $U$.

This formulation gives the function $Q$ a clear economic meaning. For an agent of type $u$ at time $t$, $Q(t,u)$ measures the sensitivity of the optimal contract with respect to the agent's output. Thus, $Q(t,u)$ is the slope of the performance-based component of the contract and coincides with the equilibrium effort induced by the contract. As a consequence of Theorem \ref{Lemma.qualitative}, more influential agents receive steeper incentives. This is consistent with the standard moral-hazard intuition that contract slopes should reflect the value of incentivizing an agent's performance. In the present setting, that value depends on the agent's influence through the interaction structure.

    \subsubsection{Spectral analysis in symmetric interactions}

We next consider the benchmark case in which the interaction function is symmetric:
\[
G(u,v)=G(v,u), \qquad (u,v)\in[0,1]^2.
\]
In this case, under Assumption \ref{ass.G}, the integral operator
\[
Kf(u)=\int_0^1 G(v,u)f(v)\,dv, \qquad f\in L^2([0,1]),
\]
is compact and self-adjoint. Since the effort process satisfies the backward linear equation
\[
 \dot{Q}(t,\cdot)=-KQ(t,\cdot), \qquad Q(T,\cdot)=\mathbbm{1}_{[0,1]}(\cdot),
\]
we can diagonalize the dynamics using the eigenbasis of $K$.
Since $K$ is compact and self-adjoint, choose an orthonormal basis $(\varphi_i)_{i= 1}^\infty$ of $L^2\left([0,1] \right)$
consisting of eigenfunctions of $K$, with corresponding eigenvalues $(\lambda_i)_{i=1}^\infty$. Set
\[
c_i:=\int_0^1\varphi_i(v)\,dv, \qquad m_i:=\int_{[0,1]\times\RR}x\varphi_i(u)\lambda(du,dx).
\]
Then the optimal effort is
\[
\alpha^*(t,u)=Q(t,u)
=
\sum_{i=1}^{\infty} c_i e^{\lambda_i(T-t)}\varphi_i(u).
\]
Substituting this expression into Theorem \ref{opt.contract.interaction.example.1} gives the following closed-form expression for the principal's value:
\[
V_p
=
\sum_{i=1}^{\infty} c_i e^{\lambda_iT}m_i
+
\frac12\sum_{i=1}^{\infty}c_i^2\int_0^T e^{2\lambda_i(T-t)}\,dt
-\int_0^1R_a(u)\,du.
\]
This formula shows that the optimal contract depends not only on an agent's
individual type \(u\). Instead, the principal rewards directions in which
effort has the largest effect on the population as a whole. The interaction
operator \(K\) identifies these directions. Its eigenfunctions represent
different patterns of effort across agents, while its eigenvalues measure how
strongly each pattern is reinforced by the interaction structure.

The spectral representation shows how the timing of incentives depends on the eigenvalues of the interaction operator. If \(\lambda_i>0\), then
$$
e^{\lambda_i(T-t)}
$$
is larger for earlier times \(t<T\). Hence, the corresponding eigenmode is weighted more heavily earlier in the contract horizon. Economically, a positive eigenvalue measures how strongly effort along the corresponding eigenfunction is amplified by the interaction network. Incentivizing such a mode has effects that propagate through the network rather than remaining local. The principal therefore assigns stronger incentives to eigenmodes associated with larger positive eigenvalues.

By contrast, the coefficient
$$
c_i = \int_0^1 \varphi_i(v)\,dv
$$
captures a different consideration: whether the eigenmode matters for aggregate
output. Some patterns of effort may be strongly amplified by the network but
still have little effect on the principal's objective because their positive and
negative components cancel out across the population. These modes are
orthogonal, or nearly orthogonal, to the constant function, so \(c_i\) is small.
As a result, they receive little weight in the optimal effort profile.

The most important modes are therefore those that satisfy two conditions at
once. First, they must affect aggregate output, meaning that \(c_i\) is large.
Second, they must be reinforced by the interaction structure, meaning that
\(\lambda_i\) is large and positive. The optimal contract puts the strongest
incentives on precisely these modes: patterns of effort that both contribute to
the principal's objective and spread effectively through the population.

\subsubsection{Comparative statics}
        In this section, we develop a deeper understanding of how optimal incentives depend on agents' influence on one another.
        We show that, from the principal's viewpoint, it is optimal to incentivize more influential agents to exert higher effort. To do so, we formalize the notion of influence.
        \begin{definition}
        Let $u_1, u_2 \in [0,1]$ be two agent types.
        We say that $u_1$ is at most as influential as $u_2$ if and only if
        \begin{align*}
        G(v,u_1) &\leq G(v,u_2),\quad \forall v \in [0,1].        \end{align*}
        \end{definition}
        The following theorem shows that agents with greater influence are incentivized to exert higher effort.
        \begin{theorem}\label{Lemma.qualitative}
            Let $G:[0,1]\times [0,1]\to \RR$ be a nonnegative interaction function that satisfies Assumption \ref{ass.G}, and $u_1, u_2 \in [0,1]$.
            Suppose that $u_1$ is at most as influential as $u_2$. Then, $\alpha^*(t,u_1) \leq \alpha^*(t,u_2)$, for all $t \in [0,T]$.
        \end{theorem}
        \begin{proof}
            Using Theorem \ref{opt.contract.interaction.example.1}, we have $\alpha^*(\cdot,u_i)  = Q(\cdot,u_i)$ for $i \in \{1,2\}$, where $Q$ solves the linear system
            \begin{align*}
                \dot{Q}(t,u) &=-\int_0^1G(v,u)Q(t,v)dv,\\
                Q(T,u) &= \mathbbm{1}_{[0,1]}(u).
            \end{align*}
           This system can be rewritten as
            \begin{align*}
                Q(t,u) = 1+\int_t^T\int_0^1 G(v,u)Q(s,v)dv\,ds.
            \end{align*}
            Let $K$ be the positive integral operator $(Kf)(u):=\int_0^1G(v,u)f(v)dv$. The terminal-value problem gives $Q(t,\cdot)=e^{(T-t)K}\mathbbm{1}_{[0,1]}(\cdot)$, and hence
            \[
            Q(t,\cdot)=\mathbbm{1}_{[0,1]}(\cdot)+\sum_{n\geq 1}\frac{(T-t)^n}{n!}K^n\mathbbm{1}_{[0,1]}(\cdot)\geq \mathbbm{1}_{[0,1]}(\cdot),
            \]
            because $G$ is nonnegative and therefore $K$ preserves nonnegative functions. Then,
            \begin{align*}
              \alpha^*(t,u_1) &=  Q(t,u_1) \\
              &= 1 + \int_t^T\int_0^1 G(v,u_1)Q(s,v)dv\,ds \\
              &\leq  1 + \int_t^T\int_0^1 G(v,u_2)Q(s,v)dv\,ds \\
              &= \alpha^*(t,u_2),
            \end{align*}
         which proves the statement.
            \qed
        \end{proof}

\paragraph{Source value density.}
The expression for the principal's value in Theorem~\ref{opt.contract.interaction.example.1} can be decomposed by source type. Let
$
m_0(u):=\int_{\mathbb R}x\,\widehat\lambda(u,dx)
$. Then the principal's value satisfies
\[
V_p
=
\int_0^1
\left[
Q(0,u)m_0(u)
+
\frac12\int_0^T Q(t,u)^2\,dt
-
R_a(u)
\right]du.
\]
We therefore define the source value density by
\begin{align}\label{source_value}
v_p(u)
:=
Q(0,u)m_0(u)
+
\frac12\int_0^T Q(t,u)^2\,dt
-
R_a(u).
\end{align}
Thus \(V_p=\int_0^1v_p(u)du\). The density \(v_p\) decomposes aggregate value by the source of value creation in the interaction structure. The first term in \eqref{source_value} is the value of the initial state of type \(u\), propagated through the network by the adjoint coefficient \(Q(0,u)\). The second term is the net value generated by the effort optimally induced from type \(u\), since \(\alpha^*(t,u)=Q(t,u)\). The last term subtracts the reservation utility required to make participation acceptable for type \(u\).

This source-based density is different from the conditional payoff
\[
\EE^{\PP^{\alpha^*,\mu^*}}[X_T-\xi^*\mid U=u],
\]
which attributes value to the type at which terminal output is realized. By contrast, the source density measures the contribution of each type's induced effort to the principal's value. It is therefore the natural object for studying incentive design: higher outward influence makes a type's effort more valuable to the principal and leads to steeper optimal incentives. This idea is formalized in the following theorem.

\begin{theorem}\label{Theorem.vp.influence}
    Let $G:[0,1]\times[0,1]\to\RR$ be a nonnegative interaction function satisfying Assumption \ref{ass.G}, and let $u_1,u_2\in[0,1]$. Suppose that $u_1$ is at most as influential as $u_2$. If $m_0(u_1)=m_0(u_2)\ge0$ and $R_a(u_1)=R_a(u_2)$, then
    \begin{align*}
        v_p(u_1)\le v_p(u_2).
    \end{align*}
\end{theorem}

\begin{proof}
    By Theorem~\ref{Lemma.qualitative}, $Q(t,u_1)\le Q(t,u_2)$ for every $t\in[0,T]$. Since $Q\ge1$, the inequality is preserved after multiplying the time-zero term by the common nonnegative initial mean, and after squaring and integrating the continuation term. The formula \eqref{source_value} gives $v_p(u_1)\le v_p(u_2)$. \qed
\end{proof}

The theorem isolates the comparative static when the initial distribution and the reservation utility are the same across two agents. A more influential agent contributes more to the principal's value because that agent's effort has a larger effect on aggregate terminal output. Thus, the influence ordering is reflected both in the optimal effort $\alpha^*$ and in the source value density $v_p$.

\subsubsection{Numerical results}

We use numerical simulations to illustrate how the interaction function determines the agents' optimal efforts and the principal's value. By Theorem~\ref{opt.contract.interaction.example.1}, the induced effort is
\[
    \alpha^*(t,u)=Q(t,u),
    \qquad (t,u)\in[0,T]\times[0,1].
\]
The same object is therefore both the equilibrium effort and the slope of the performance-based component of the contract. In the figures below, we report the interaction function $G$, the outward influence profile
\[
    C_G(u):=\int_0^1G(v,u)dv,
\]
the source value density $v_p$ defined in \eqref{source_value}, and selected profiles $u\to Q(t,u)$ at four times. In all numerical examples we set $T=1$, $R_a=0$, and $\hat\lambda(u,dx)=\delta_0(dx)$. Hence the initial-state term vanishes and differences in $v_p$ and $V_p=\int_0^1v_p(u)du$ are entirely generated by the incentive term $\frac12\int_0^T Q(t,u)^2dt$.

The four examples are chosen to isolate distinct economic structures: reciprocal local interaction, a global hierarchy, an extreme core--periphery organization, and separated teams with different within-team hierarchical intensities. Each interaction function is normalized to have fixed aggregate outward influence $\int_0^1C_G(u)=1$.

We focus on nonnegative interaction functions $G$. A useful benchmark for interpreting the figures is the elementary lower bound
\[
    Q(t,u)\geq 1+(T-t)C_G(u).
\]
Indeed, the adjoint representation gives
\[
    Q(t,u)=1+\int_t^T\int_0^1G(v,u)Q(s,v)dv\,ds.
\]
Since $G\geq0$ and $Q(s,v)\geq1$, the integral is bounded below by
\[
    Q(t,u)\geq 1+ \int_t^T\int_0^1G(v,u)dv\,ds=1+(T-t)C_G(u).
\]

This explains why dispersion in $C_G$ tends to translate into dispersion in $Q$, and hence into larger dispersion in the principal's source value $v_p$.

\begin{enumerate}
    \item \textit{Reciprocal local interactions. $G(u,v)
	:=
	b
	+
	a e^{-\rho |u-v|}
	+
	d uv$.}

	The first interaction function is a local benchmark. Nearby types interact most strongly, and the term $duv$ gives higher types somewhat greater average connectivity. This example has moderate dispersion in outward influence. We set $
	a=0.65, b=0.10, \rho=10, d=0.25$.

    \begin{figure}[H]
        \centering
        \includegraphics[width=0.9\textwidth]{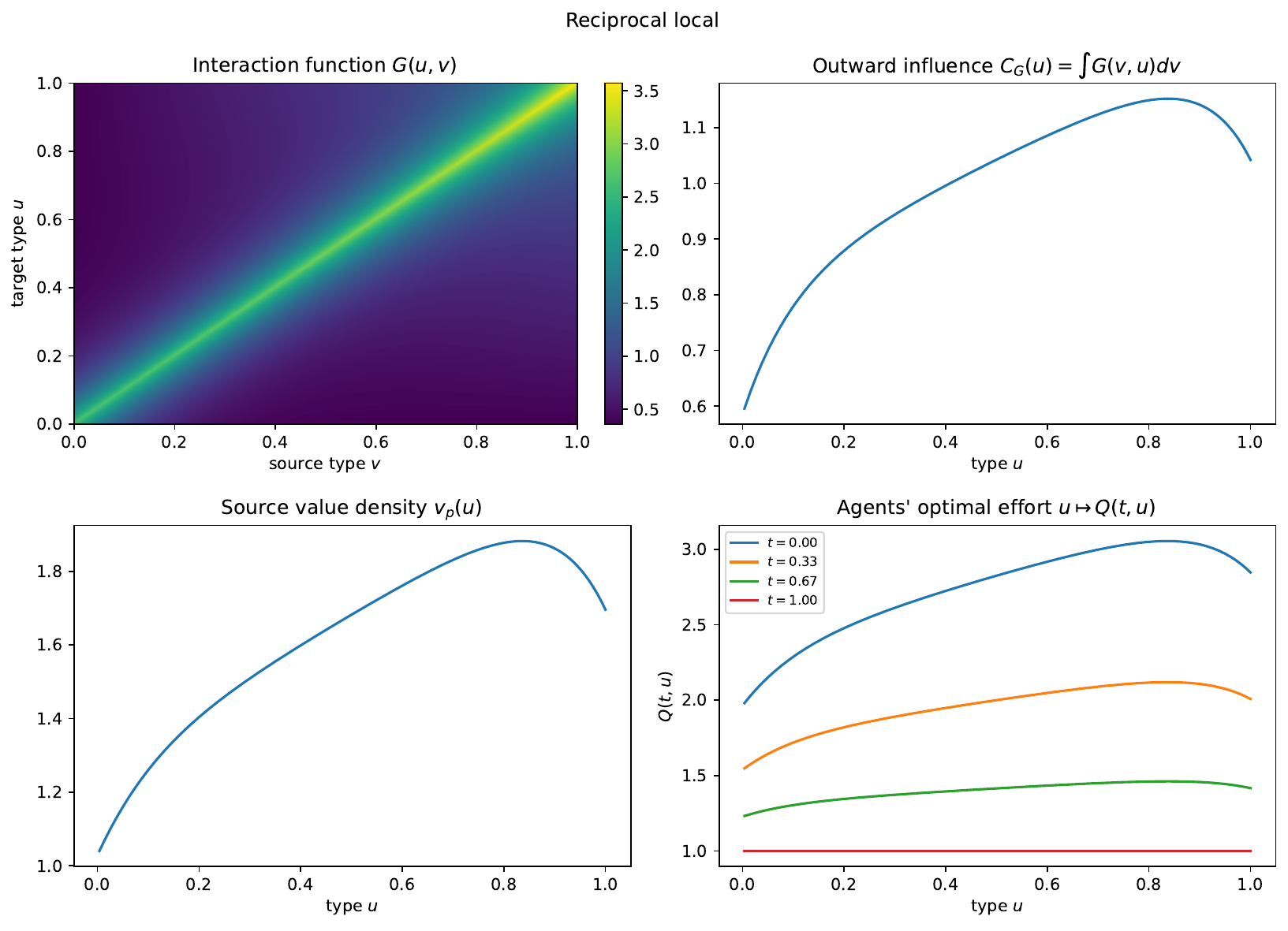}
        \caption{Reciprocal local interactions. The panels report the normalized interaction function $G$, outward influence $C_G$, source value $v_p$, and selected agents' optimal efforts $u\to Q(t,u)$.}
        \label{figureG1}
    \end{figure}

    Figure~\ref{figureG1} provides the baseline. Because the interaction is largely reciprocal and local, incentives do not concentrate on a small set of agents. Still, types with larger outward influence contribute more to the principal's value.

    \item \textit{Global hierarchy. $G(u,v):=\frac{1}{1+\exp\{\theta(u-v)\}}$.}

    The second interaction function is directional.  Higher-ranked agents affect a large mass of lower-ranked agents, while lower-ranked agents have little influence in return. We set $\theta=14$.

    \begin{figure}[H]
        \centering
        \includegraphics[width=0.9\textwidth]{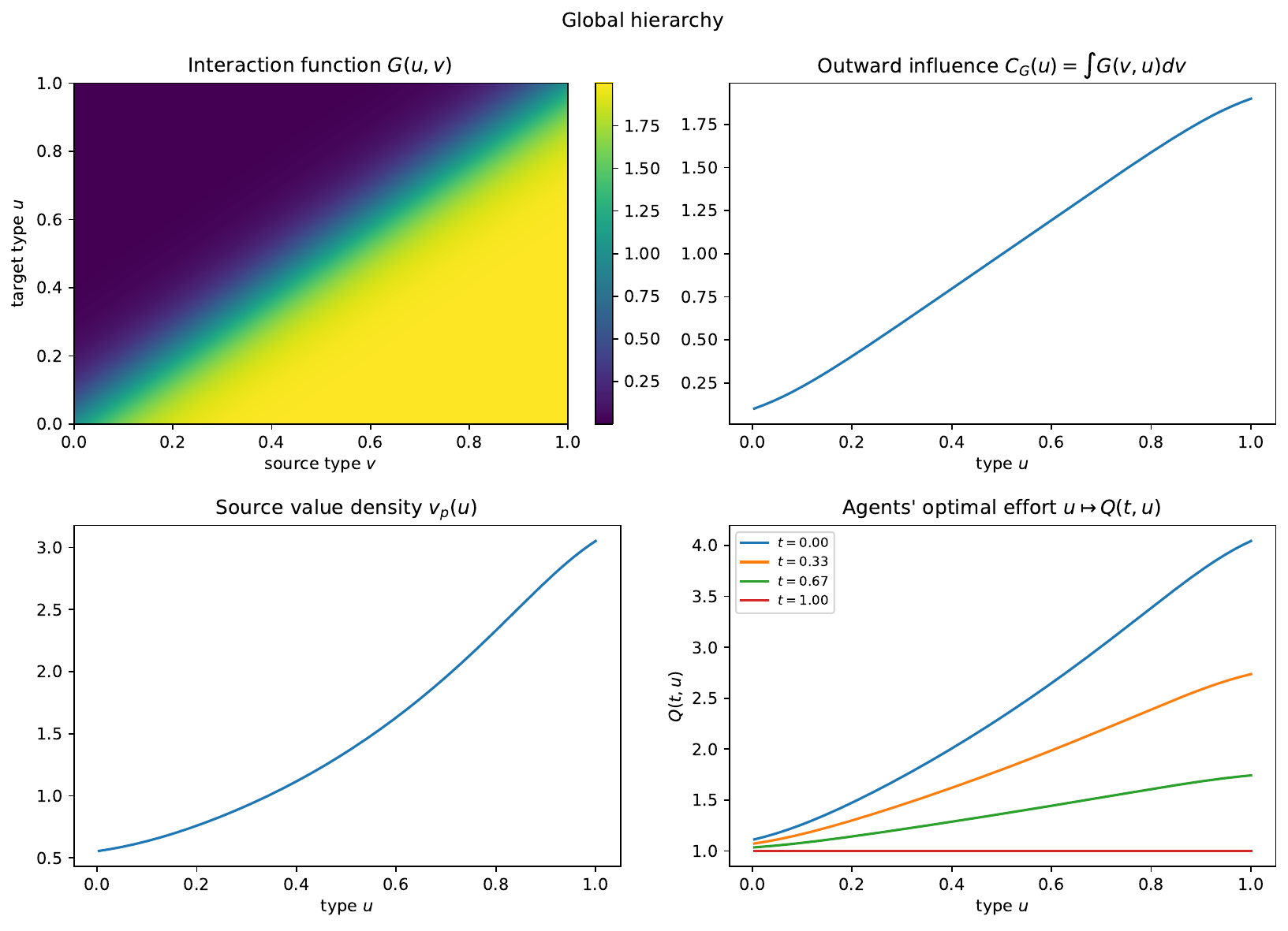}
        \caption{Global hierarchy. The interaction function is normalized. Upstream agents have larger outward influence, higher source value, and steeper optimal incentives.}
        \label{figureG2}
    \end{figure}

    Figure~\ref{figureG2} has a clear organizational interpretation. The top of the hierarchy receives the largest effort coefficient because effort by upstream agents propagates through a large part of the population. Accordingly, upstream agents receive larger incentives and contribute more to the principal's source value.

    \item \textit{Core--periphery interactions. $G(u,v):=
	a
	\left(
	\frac{1}{1+\exp\{-\kappa(v-c)\}}
	\right)
	\left(
	\ell + m e^{-\eta u}
	\right)+b$.}

	 The third interaction function makes the core--periphery structure sharper: a very small core carries most of the outward influence, while peripheral agents have only weak spillover effects. We set $
    a=0.92, b=0.08, c=0.72,
	 \kappa=30,
	 \ell =0.30,
	 m =0.70,
	 \eta =1.7$.

    \begin{figure}[H]
        \centering
        \includegraphics[width=0.9\textwidth]{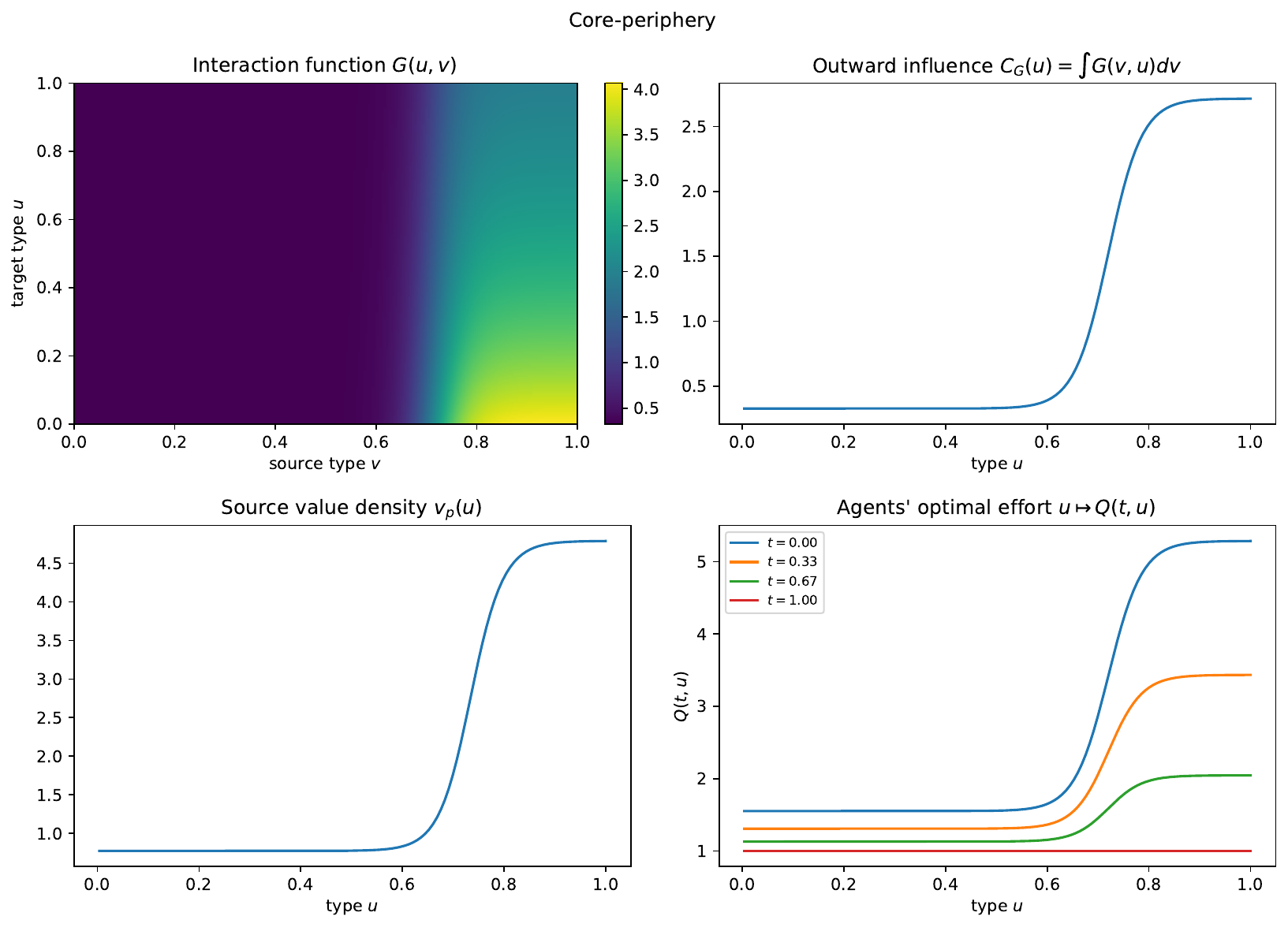}
        \caption{Core--periphery interactions. The interaction function is normalized. A small core generates very large spillovers and the optimal contract assigns substantially higher efforts and steeper incentive slopes to that core.}
        \label{FigureG3}
    \end{figure}

    Figure~\ref{FigureG3} shows that incentives need not rise smoothly along a global ranking. What matters is outward influence. In this core--periphery case, influence is concentrated on a small set of core types. As a result, the principal concentrates the strongest incentives on the core because these agents generate the largest marginal effect on the rest of the system. This example produces visibly larger variation in the principal's source value than the smoother benchmarks.

    \item \textit{Team hierarchies with different intensities. $G(u,v):=\sum_{i=0}^{L-1}
	a_i
	\mathbbm{1}_{ \left(\frac{i}{L},\frac{i+1}{L}\right]^2}(u,v)
	\frac{1}{
		1+\exp\left\{
		L\beta_i(u-v)
		\right\}
	}
		.$}

		The fourth interaction function divides the population into four separated teams. There are no cross-team interactions, but each team has its own internal hierarchy. The hierarchical intensity and interaction strength differ across teams, so the figure combines local hierarchy with heterogeneous organizational units. We set $
		L=4, (a_0,a_1,a_2,a_3)= (0.55,0.85,1.20,1.60),
		(\beta_0,\beta_1,\beta_2,\beta_3)
		=(6,10,15,22)$. In the convergence experiment below, all four benchmarks use the same dyadic partitions.

    \begin{figure}[H]
        \centering
        \includegraphics[width=0.9\textwidth]{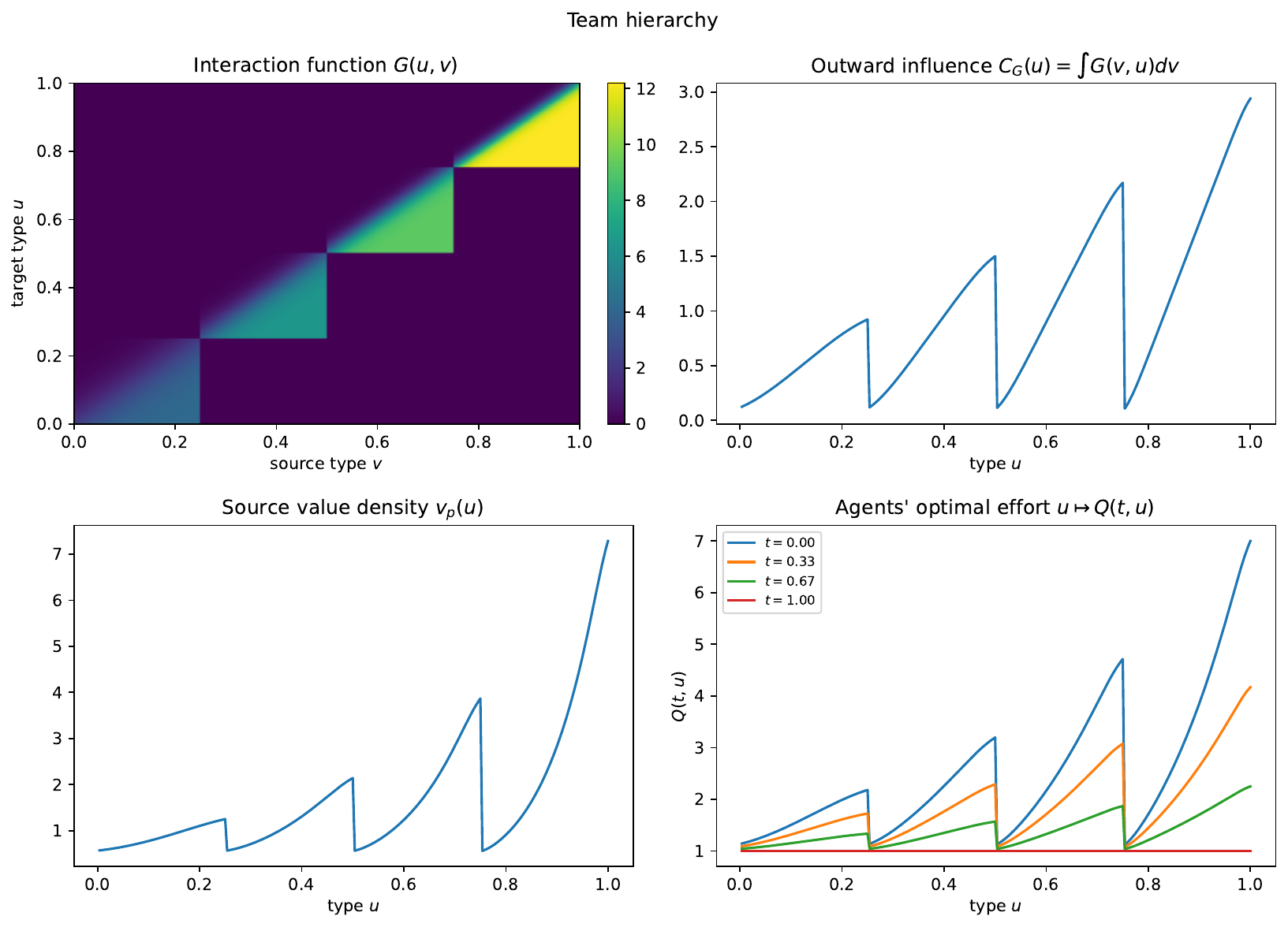}
        \caption{Team hierarchies with different intensities. The interaction function is normalized. Teams are separated, and each team has a different internal hierarchy strength. Upstream agents receive larger incentives locally, with sharper variation in the more hierarchical teams.}
        \label{FigureG4}
    \end{figure}
\end{enumerate}

Figure~\ref{FigureG4} shows why it is useful to allow different team intensities. The same qualitative pattern appears in every team: upstream agents receive larger incentives. The strength of this effect, however, varies by team. Teams with stronger internal hierarchy display sharper dispersion in $C_G$, $Q_G$, and $v_p$. This shows how the optimal contract reacts to local organizational intensity rather than only to the presence of teams.

Taken together, the four examples show that the optimal contract is shaped by the dispersion of outward influence. Reciprocal networks generate relatively smooth incentives; global hierarchies reward upstream positions; core--periphery networks reward the high-impact core even more sharply; and team hierarchies reward upstream agents locally, with the steepness depending on the intensity of each team hierarchy.

\paragraph{Dispersion of influence and principal value.}
Table~\ref{tab:variance-C-vp} reports the variance of the outward influence profile $C_G$ and the variance of the type-level value contribution $v_p$. The comparison focuses on the relation between dispersion of influence and the principal's aggregate value, rather than on the marginal-value profile alone. Using the elementary bound $Q(t,u)\geq 1+(T-t)C_G(u)$, we obtain, in our benchmark cases,
\begin{align}\label{lower_bound}
    V_p
    =\frac{1}{2}\int_0^T\int_0^1 Q(t,u)^2dudt\nonumber
    &\geq
    \frac{T}{2}+\frac{T^2}{2}\int_0^1C_G(u)du+\frac{T^3}{6}\int_0^1C_G(u)^2du \\
    &= \frac{T}{2}+\frac{T^2}{2}\int_0^1C_G(u)du+\frac{T^3}{6}\left(\int_0^1C_G(u)du\right)^2+\frac{T^3}{6}\operatorname{Var}\left( C_G(U)\right).
\end{align}
For fixed aggregate outward influence $\int_0^1C_G(v)dv=\int_0^1\int_0^1 G(u,v)dudv$, the last term in \eqref{lower_bound} increases with $\operatorname{Var}(C_G)$. Thus, greater dispersion of outward influence pushes up the lower bound for the principal's value. In the following table, we compare the principal's value with the variances of the source value density and outward influence for the four benchmark interactions studied previously. We normalize the corresponding interaction functions to impose $\int_0^1C_{G}(v)dv=1$ in each example. 

\begin{table}[H]
\centering
\caption{Dispersion of outward influence and source value density across the four numerical examples.}
\label{tab:variance-C-vp}
\begin{tabular}{lccc}
\toprule
Example & $\operatorname{Var}(C_G)$ & $\operatorname{Var}(v_p)$ & $V_p$\\
\midrule
Reciprocal local & 1.997e-02 & 5.287e-02 & 1.621e+00\\
Global hierarchy & 3.059e-01 & 5.816e-01 & 1.526e+00\\
Core-periphery & 9.695e-01 & 2.665e+00 & 1.844e+00\\
Team hierarchy & 4.891e-01 & 1.896e+00 & 1.685e+00\\
\bottomrule
\end{tabular}
\end{table}

The computed values should not be interpreted as establishing a strict ranking across all interaction structures, since the principal's value depends on the full shape of $G$, not only on $\operatorname{Var}(C_G)$. Nevertheless, the examples illustrate the general mechanism: concentrating influence on high-impact agents can strengthen network amplification of effort and increase the value extracted from performance-based incentives.

\subsubsection{Stability results}
		Another relevant question is whether optimal contracts are stable under changes in the interaction function. The following theorem estimates how agents' optimal efforts change under perturbations of ${G}$. For an interaction function ${G}: [0,1]^2 \to \RR$, denote by $\alpha_G^*$ the agents' optimal effort under ${G}$, and by $\mu^*_G$ the agents' equilibrium distribution under $G$. By Theorem \ref{opt.contract.interaction.example.1}, $\alpha_G^*(t,u)=Q_G(t,u)$, where $Q_G$ is the solution of the system \eqref{infinite.system}. 
        \begin{theorem}\label{lemma.Q.stability}
        Let $G_1, G_2:[0,1]^2 \to \RR$ be two interaction functions satisfying Assumption \ref{ass.G}. Then there exists a constant $C>0$, depending on $\|G_1\|_{\infty},\|G_2\|_{\infty}$, and $T$, such that
			\begin{align*}
        \sup_{u \in [0,1]}\sup_{t\in [0,T]}|\alpha^*_{G_1}(t,u)-\alpha^*_{G_2}(t,u)|=    &\sup_{u \in [0,1]}\sup_{t\in [0,T]}|Q_{G_1}(t,u)-Q_{G_2}(t,u)|\leq C\|G_1-G_2\|_{\infty}
          \end{align*}
		\end{theorem}
        \begin{proof}
            We start by recalling that $Q_{G_1}$ and $Q_{G_2}$ solve the systems, respectively,
        \begin{align*}
                Q_{G_1}(t,u) &= 1+\int_t^T\int_0^1G_1(v,u)Q_{G_1}(s,v)dv\,ds,\\   Q_{G_2}(t,u) &= 1+\int_t^T\int_0^1G_2(v,u)Q_{G_2}(s,v)dv\,ds,
            \end{align*}
            for all $(t,u) \in [0,T]\times [0,1]$.
            Adding and subtracting $\int_t^T\int_0^1G_1(v,u)Q_{G_2}(s,v)dv\,ds$ yields
            \begin{align*}
                \sup_{u\in [0,1]}|Q_{G_1}(t,u)-Q_{G_2}(t,u)| &\leq \Big(T\sup_{(t,u) \in [0,T]\times [0,1]}|Q_{G_2}(t,u)| \Big)\|G_1-G_2\|_{\infty}\\
&+\left(\|G_1\|_{\infty}+\|G_2\|_{\infty}\right)\int_t^T\sup_{u\in [0,1]}|Q_{G_1}(s,u)-Q_{G_2}(s,u)|ds, \quad t\in [0,T].
            \end{align*}
            Thus, by Grönwall's inequality, there exists $C>0$ (depending on $\|G_1\|_{\infty}$, $\|G_2\|_{\infty}$, and $T$) such that
            \begin{align*}
                \sup_{u \in [0,1]} \sup_{t\in [0,T]} |Q_{G_1}(t,u)-Q_{G_2}(t,u)|\leq C\|G_1-G_2\|_{\infty}.
            \end{align*}
            This completes the proof.
            \qed
        \end{proof}
        
       Using Theorem \ref{lemma.Q.stability}, we show that the equilibrium distribution of the optimal contracts varies continuously with the agents' interactions. To measure the distribution of each agent's contract in equilibrium, for each interaction function \(G\), we define the probability kernel
       \(K_G:[0,1]\to \mathcal{P}^2(\RR)\) by
       \begin{equation}\label{kernel.contract}
       K_G(u,\cdot)
       :=
       \mathcal N\left(
       R_a(u)+\frac12\int_0^T Q_G(t,u)^2\,dt,\,
       \int_0^T Q_G(t,u)^2\,dt
       \right).
       \end{equation}
       By Theorem \ref{opt.contract.interaction.example.1}, $K_G$ is a version of the regular conditional law of \(\xi_G^*\)
       given \(U=u\) under \(\mathbb P^{Q_G(\cdot,U),\mu_G^*}\). 
       The following theorem formalizes the result.
		\begin{theorem}\label{theorem.contracts.stability}
            Suppose Assumption \ref{ass.reservation} holds. Let $G_1, G_2:[0,1]^2 \to \RR$ be two interaction functions satisfying Assumption \ref{ass.G}. There exists a constant $C>0$, depending on $G_1$ and $G_2$, such that
            \[
            \mathcal{W}_2\left(K_{G_1}(u,\cdot),K_{G_2}(u,\cdot)\right)
            \leq C\left(\|G_1-G_2\|^{1/2}_{\infty}+\|G_1-G_2\|_{\infty}\right),
            \]
            for all $u\in [0,1]$. Here $\xi^*_{G_i}$ denotes the optimal contract defined in Theorem \ref{opt.contract.interaction.example.1} with $G_i$, for $i \in \{1,2\}$.
			\end{theorem}
		\begin{proof}
         Using the closed-form formula for the $\mathcal{W}_2$-distance between two Gaussian distributions, we obtain
          \begin{align*}
          \mathcal{W}_2\left(K_{G_1}(u,\cdot),K_{G_2}(u,\cdot)\right)^2
          &=  \frac{1}{4}\left(\int_0^T\left( Q_{G_1}(t,u)^2-Q_{G_2}(t,u)^2\right)dt\right)^2\\
          &\quad +\left(\left(\int_0^T Q_{G_1}(t,u)^2dt\right)^{1/2}-\left(\int_0^TQ_{G_2}(t,u)^2dt\right)^{1/2}\right)^2.
          \end{align*}
          Using the $1/2$-Hölder continuity of $\sqrt{x}$ and Theorem \ref{lemma.Q.stability}, there exists a constant $C>0$ (depending on $G_1$ and $G_2$) such that
          \[
          \mathcal{W}_2\left(K_{G_1}(u,\cdot),K_{G_2}(u,\cdot)\right)
          \leq C\left(\|G_1-G_2\|_{\infty}^{1/2}+\|G_1-G_2\|_{\infty}\right).
          \]
          This completes the proof.
		    \qed
		\end{proof}

    \subsection{Convergence results}\label{subsectionconvergence}

We start by showing that, as the number of agents grows, the equilibrium effort chosen by each agent in the
$N$-agent problem converges to the heterogeneous mean-field equilibrium. Recall that, for all $N\in \NN$ and $i \in [N]$, $\alpha^{i,N,*}$ and $Q^{i,N}$ are defined in Theorem \ref{opt.contract.finite.example1}.
\begin{theorem}\label{lemmaQ}
Suppose Assumption \ref{ass.G} holds. Then there exists a constant $C>0$, independent of $N$, such that, for all $N\in \mathbb{N}$,
		$$ \max_{i \in [N]}\sup_{t\in [0,T]}\big|\alpha^*(t,i/N)-\alpha^{i,N,*}(t)\big| = \max_{i \in [N]}\sup_{t\in [0,T]}\bigg|Q\left(t,\frac{i}{N}\right)-Q^{i,N}(t)\bigg|\leq \frac{C}{N}.$$
	\end{theorem}
	\begin{proof}
		For $N\in \NN$, $i \in [N]$, and $t \in [0,T]$, adding and subtracting terms gives
		\begin{align}\label{ineq.Q}  &\big|Q^{i,N}(t)-Q(t,i/N)\big|
			\\
			&\leq  \int_t^T \bigg|\frac{1}{N}\sum_{j=1}^NG\left(\frac{j}{N},\frac{i}{N}\right)\left(Q\left(s,\frac{j}{N}\right)-Q^{j,N}(s)\right) \bigg|ds \nonumber\\
		  &+\int_t^T \bigg|\int_0^1G\left(v,\frac{i}{N}\right)Q(s,v)dv-\frac{1}{N}\sum_{j=1}^NG\left(\frac{j}{N},\frac{i}{N}\right)Q\left(s,\frac{j}{N}\right)\bigg|ds\nonumber\\
			&=:  I^{i,N}_1(t)+  I^{i,N}_2(t),\nonumber
		\end{align}
		where
		\begin{align*}
			I^{i,N}_1(t) &:= \int_t^T \bigg|\frac{1}{N}\sum_{j=1}^NG\left(\frac{j}{N},\frac{i}{N}\right)\left(Q\left(s,\frac{j}{N}\right)-Q^{j,N}(s)\right) \bigg|ds, \\
			I^{i,N}_2(t) &:=\int_t^T \bigg|\int_0^1G\left(v,\frac{i}{N}\right)Q(s,v)dv-\frac{1}{N}\sum_{j=1}^NG\left(\frac{j}{N},\frac{i}{N}\right)Q\left(s,\frac{j}{N}\right)\bigg|ds.
		\end{align*}
By Assumption \ref{ass.G} and the piecewise Lipschitz regularity of $Q(s,\cdot)$, the functions
$v \mapsto G\left(v,\frac{i}{N} \right)Q(s,v)$
are uniformly bounded and piecewise Lipschitz on the common finite partition, uniformly over $s \in [0,T]$ and $i\in [N]$. Therefore, the Riemann sums have
error $\mathcal{O}(1/N)$, uniformly in $(s,i) \in [0,T]\times [N]$.
		Thus, there exists $C_2>0$ such that

		$$
		\sup_{0\leq t\leq T}\big|I_2^{i,N}(t)\big| \leq \frac{C_2}{N}.
		$$
		Taking the supremum in \eqref{ineq.Q} over $i \in [N]$, we obtain
		\begin{align*}
			\sup_{i \in [N]}\bigg|Q\left(t,\frac{i}{N}\right)-Q^{i,N}(t) \bigg| &\leq \frac{C_2}{N}+\|G\|_{\infty}\int_t^T\sup_{j\in [N]}\left|Q\left(s,\frac{j}{N}\right)-Q^{j,N}(s) \right|ds, \quad t\in [0,T].
\end{align*}
Using Grönwall's inequality, there exists $C>0$, independent of $N$, such that
		\begin{align*}
			\max_{i \in [N]}\sup_{t\in [0,T]}\bigg|Q\left(t,\frac{i}{N}\right)-Q^{i,N}(t) \bigg|\leq \frac{C}{N}.
		\end{align*}
		This proves the theorem.
		\qed
	\end{proof}
	Next, we show that the principal's value function and the optimal contracts for the $N$-agent problem converge to the heterogeneous mean-field solution. We introduce a sequence of mutually independent random variables $\Big(\big(X_0^{i,N}\big)_{i\in [N]}\Big)_{N\in \mathbb{N}}$ satisfying $X_0^{i,N}\sim \hat{\lambda}\left(i/N,dx\right)$ for all $i \in [N]$. In the following theorem, we assume that the initial value of agent $i$'s project in the $N$-agent problem is given by $X_0^{i,N}\sim \hat{\lambda}(i/N,dx)$. With a slight abuse of notation in Theorem \ref{conv.Value}, we write \(V_p^N\) for the finite-agent value obtained by evaluating the
	formula in Theorem \ref{opt.contract.finite.example1} at the random initial vector $(X_0^{i,N})_{i\in [N]}$. Using the same notation as in Theorem \ref{theorem.contracts.stability}, we denote by
	$K_G$ 
	the version of the regular conditional law of $\xi^*_G$, given $U=u$, under $\PP^{\alpha^*_G,\mu^*_G}$ defined in \eqref{kernel.contract}.
    \begin{theorem}\label{conv.Value}
		Suppose Assumptions \ref{ass.G} and \ref{ass.reservation} hold.
	Then, there exists $C>0$, independent of $N$, such that
		\begin{align*}	 \EE^{\lambda}\left[(V^N_p - V_p)^2\right]&\leq \frac{C}{N}, \quad
\sup_{i\in [N]}\mathcal{W}_2\left(K_G\left(\frac{i}{N},\cdot\right),\PP^{\bar{Q}^N}\circ \left(\xi^{i,N,*}\right)^{-1}\right)\leq \frac{C}{\sqrt{N}},
		\end{align*}
        for all $N \in \mathbb{N}$.
	\end{theorem}

	\begin{proof}
		For $N \in \mathbb{N}$, adding and subtracting terms gives
		\begin{align}\label{eqn.VPN}
			V_p^{N}-V_p&= \frac{1}{N}\sum_{i=1}^N \left(Q^{i,N}(0)-Q(0,i/N)\right)X^{i,N}_0+\frac{1}{N}\sum_{i=1}^NQ(0,i/N)\left(X_0^{i,N}-\int_{\RR}x\hat{\lambda}(i/N,dx)\right) \nonumber \\
			&+ \frac{1}{N}\sum_{i=1}^N Q(0,i/N)\int_{\RR} x\hat{\lambda}(i/N,dx)-\int_{0}^1Q(0,u)\int_{ \mathbb{R}}x\hat{\lambda}(u,dx)du\nonumber \\
            &+\frac{1}{2N}\sum_{i=1}^N\int_0^T\left(Q^{i,N}(t)\right)^2dt -\frac{1}{2N}\sum_{i=1}^N\int_0^TQ(t,i/N)^2dt\nonumber\\
            &+\frac{1}{2N}\sum_{i=1}^N \int_0^TQ(t,i/N)^2dt- \frac{1}{2}\int_0^T\int_0^1Q(t,u)^2dudt\nonumber\\
            &-\frac{1}{N}\sum_{i=1}^N R_a(\frac{i}{N})+\int_{0}^1R_a(u)du\nonumber \\
            &=:\sum_{i=1}^6{I}^N_i.
		\end{align}

		Applying the Cauchy-Schwarz inequality and Theorem \ref{lemmaQ}, there exists a constant $C_1>0$ such that
		\begin{align}
        \label{bound1.theorem.Value}
			\EE^{\lambda}\left[|I_1^N|^2\right]&= \EE^{\lambda}\left[\bigg| \frac{1}{N}\sum_{i=1}^N\left(Q^{i,N}(0)-Q(0,{i/N}) \right)X_0^{i,N}\bigg|^2\right]\nonumber \\
            &\leq \max_{i\in [N]}\sup_{0 \leq t\leq T}\big|Q^{i,N}(t)-Q(t,i/N)\big|^2 \sup_{u\in [0,1]}\int_{\RR}|x|^2\hat{\lambda}(u,dx)\\
            &\leq \frac{C_1}{N}.\nonumber
		\end{align}
        Since $\left(X_0^{i,N}\right)_{i=1}^N$ are independent, Assumption \ref{ass.G} implies that there exists a constant $C_2>0$, independent of $N$, such that
        \begin{align}        \EE^{\lambda}\left[|I_2^N|^2\right]&=\EE^{\lambda}\left[\left|\frac{1}{N}\sum_{i=1}^N Q(0,i/N)\left( X_0^{i,N}-\int_{\RR}x\hat{\lambda}(i/N,dx)\right) \right|^2\right] \nonumber \\
            &\leq \frac{C_2}{N}.
        \end{align}
        By Assumption \ref{ass.G}, there exists a constant $C_3>0$, independent of $N$, such that
        \begin{align*}
        |I_3^N|  =\left|\frac{1}{N}\sum_{i=1}^N Q(0,i/N)\int_{\RR} x\hat{\lambda}(i/N,dx)-\int_{0}^1Q(0,u)\int_{ \mathbb{R}}x\hat{\lambda}(u,dx)du\right| \leq  \frac{C_3}{N},
        \end{align*}
        for $N$ sufficiently large.
        Moreover, using Assumptions \ref{ass.G} and \ref{ass.reservation}, there exists $C_4>0$ such that
        \begin{align}\label{bound.I.1.6}
\big|I^N_4\big|+\big|I^N_5\big|+\big|I^N_6\big| \leq \frac{C_4}{N}.
        \end{align}
Combining the previous inequalities, there exists $C>0$, independent of $N$, such that
\begin{align*}
    \EE^{\lambda}\left[(V_p^N-V_p)^2\right]\leq 6\sum_{i=1}^6 \EE^{\lambda}\left[\big|I_i^N \big|^2\right] \leq \frac{C}{N}.
\end{align*}
It remains to bound the $\mathcal{W}_2$-distance between the distributions of the optimal contracts. By Theorems \ref{opt.contract.finite.example1} and \ref{opt.contract.interaction.example.1}, we have
\begin{align*}
    \nu_1 &:=K_G\left(\frac{i}{N},\cdot\right) = \mathcal{N}\left(R_a(i/N)+\frac{1}{2}\int_0^TQ(t,i/N)^2dt, \int_0^TQ(t,i/N)^2dt\right), \\
    \nu_2 &:= \PP^{\bar{Q}^N} \circ \left(\xi^{i,N,*}\right)^{-1} = \mathcal{N}\left(R_a(i/N)+\frac{1}{2}\int_0^TQ^{i,N}(t)^2dt, \int_0^TQ^{i,N}(t)^2dt \right).
\end{align*}
Applying the well-known expression for the $\mathcal{W}_2$-distance between two Gaussian distributions, we obtain
\begin{align*}
    \left(\mathcal{W}_2\left(\nu_1,\nu_2 \right)\right)^2
    &= \frac{1}{4}\left(\int_0^T\left(Q(t,i/N)^2-Q^{i,N}(t)^2 \right)dt \right)^2\\
    &\quad +\left(\left(\int_0^TQ^{i,N}(t)^2dt \right)^{1/2} -\left(\int_0^TQ(t,i/N)^2dt\right)^{1/2}\right)^2 .
\end{align*}
Moreover,
\begin{align*}
\left|\int_0^T\left(Q(t,i/N)^2-Q^{i,N}(t)^2\right)dt\right|
\leq \int_0^T\left|Q(t,i/N)^2-Q^{i,N}(t)^2\right|dt.
\end{align*}
By Theorem \ref{lemmaQ} and the $\frac{1}{2}$-Hölder property of $\sqrt{x}$, there exists $C>0$, independent of $i\in [N]$, such that $\mathcal{W}_2\left(\nu_1,\nu_2 \right)\leq \frac{C}{\sqrt{N}}.$ 
This completes the proof. \qed
	\end{proof}

	The next theorem is the $N$-agent implementation result. The continuum optimizer $Q$
	defines a family of contract slopes indexed by type. Evaluating these slopes at the grid
	points $i/N$ gives a finite collection of contracts. The content of the theorem is that this
	sampled continuum contract is not only asymptotically optimal but also admissible
	in the finite-agent game: it satisfies the participation constraints, induces a unique
	Nash equilibrium, and loses only $\mathcal{O}(1/N)$ relative to the finite-agent optimum. Thus the
	continuum problem provides a practical approximation scheme for large finite contracting
	problems.
	\begin{theorem}\label{nearoptimal}
    Suppose Assumptions \ref{ass.G} and \ref{ass.reservation} hold. For each $N\in\mathbb{N}$, define the finite-agent contracts $\xi^{N,\mathrm{cont}} := (\xi^{1,N,\mathrm{cont}},\ldots,\xi^{N,N,\mathrm{cont}})$ by
    \begin{align}\label{contracts.epsilon}
        \xi^{i,N,\mathrm{cont}}
        :=&\ R_a(i/N)-\int_0^T\left(\frac{1}{2}Q(t,i/N)^2+Q(t,i/N)\frac{1}{N}\sum_{j=1}^NG\left(\frac{i}{N},\frac{j}{N}\right)X_t^{j,N}\right)dt \nonumber\\
        &\quad +\int_0^TQ(t,i/N)dX_t^{i,N}, \quad i\in [N].
    \end{align}
    The stochastic integral in \eqref{contracts.epsilon} is taken with respect to the finite-agent output process $X^{i,N}$. Then, $\xi^{N,\mathrm{cont}}$ satisfies the agents' reservation constraint, induces a unique Nash equilibrium, and there exists a constant $C>0$, independent of $N$, such that  $$V_p^N-J^N_p(\xi^{N,\mathrm{cont}})\leq \frac{C}{N}, $$
        for all $N \in \mathbb{N}$.
	\end{theorem}

    \begin{remark}

    This result is especially useful in fast-moving environments where there is little time for repeatedly solving the exact finite-agent problem. The principal may face different finite samples of agents over time, while the underlying interaction function remains relatively stable. Theorem~\ref{nearoptimal} then justifies applying the same continuum-designed contract to each sampled finite population as a scalable and quantitatively controlled approximation.
    \end{remark}

	\begin{proof}
		Consider the $N$-agent problem introduced in Subsection \ref{example1.finite.agents}. In this case, the agents' output process $\bar{X}^N$ satisfies
		\begin{align*}
			dX_t^{i,N} &=\left( \frac{1}{N}\sum_{j=1}^NG_{i,j}^NX^{j,N}_t+\alpha_t^i\right)dt +dW^{i,\alpha}_t,\quad \PP^{\alpha}-a.s.\quad i\in [N],\\
			X_0^{i,N} &=x_0^{i,N}.
		\end{align*}
		Assume the principal offers the contracts $\xi^{N,\mathrm{cont}}$ to the agents. Then every agent $i \in [N]$ maximizes the following objective:
        \begin{align*}
            J_a^i(\alpha^i,\xi^{i,N,\mathrm{cont}}) = R_a(i/N) +\EE^{\PP^\alpha}\left[\int_0^T\left(Q(t,i/N)\alpha^i_t-\frac{1}{2}(\alpha_t^i)^2-\frac{1}{2} Q(t,i/N)^2\right)dt\right].
        \end{align*}
        One verifies that there exists a unique Nash equilibrium given by
		\begin{align*}
        \hat{Q}^N(t) &:= \big(Q(t,{i}/{N})\big)_{i=1}^N, \quad t \in [0,T].
		\end{align*}
    Moreover, $V_a^{i}(\xi^{i,N,\mathrm{cont}})=R_a(i/N)$, implying $\xi^{N,\mathrm{cont}}\in \Sigma_a^N$.
     This equilibrium induces the probability measure $\PP^{\hat{Q}^N}$ defined as follows:
       \begin{align*}      \frac{d\PP^{\hat{Q}^N}}{d\WW} := \mathcal{E}\left(\int_0^\cdot \sum_{i=1}^N\left(\frac{1}{N}\sum_{j=1}^NG^N_{i,j}X_t^{j,N}+Q(t,i/N) \right)dX^{i,N}_t\right)_T.
       \end{align*}
        In equilibrium, the output process satisfies the following dynamics:
		\begin{align}\label{eqn.x.q.hat}
			dX^{i,N}_t = \left(\frac{1}{N}\sum_{j=1}^NG^N_{i,j}X^{j,N}_t+Q(t,i/N)\right)dt + dW^{i,\hat{Q}^N}_t, \PP^{\hat{Q}^N}-a.s.,\quad i\in [N].
		\end{align}
        From the principal's point of view,
\begin{align*}
			J_P^N(\xi^{N,\mathrm{cont}}) &= \EE^{\PP^{\hat{Q}^N}}\left[\frac{1}{N}\sum_{i=1}^N (X^{i,N}_T-\xi^{i,N,\mathrm{cont}}) \right] \\
            &= \EE^{\PP^{\hat{Q}^N}}\left[-\frac{1}{N}\sum_{i=1}^NR_a(i/N)+\frac{1}{N}\sum_{i=1}^NX^{i,N}_T-\frac{1}{2N}\sum_{i=1}^N\int_0^T Q(t,i/N)^2dt\right].
		\end{align*}
		Similarly, Theorem \ref{opt.contract.finite.example1} implies that
		the contracts $\bar{\xi}^{N,*}:=(\xi^{i,N,*})_{i=1}^N$ given by
		\begin{align*}
			\xi^{i,N,*} := R_a(i/N) - \int_0^T\left( \frac{1}{2}Q^{i,N}(t)^2+Q^{i,N}(t)\frac{1}{N}\sum_{j=1}^NG^N_{i,j}X^{j,N}_t\right)dt+\int_0^T{Q}^{i,N}(t)dX^{i,N}_t, \quad i \in [N],
		\end{align*}
		are optimal for the $N$-agent problem. Moreover, again by Theorem \ref{opt.contract.finite.example1}, the contracts $\bar{\xi}^{N,*}$ induce a unique Nash equilibrium of the agents given by
        \begin{align*}
            \bar{Q}^N(t):=\left(Q^{i,N}(t)\right)_{i=1}^N, \quad t \in [0,T].
        \end{align*}
        Therefore,
		\begin{align*}
			V_p^N &= J^N_p(\bar{\xi}^{N,*}) \\
            &= \EE^{\PP^{\bar{Q}^N}}\left[-\frac{1}{N}\sum_{i=1}^N R_a(i/N)+\frac{1}{N}\sum_{i=1}^N {X}_T^{i,N}-\frac{1}{2N}\sum_{i=1}^N\int_0^T \left(Q^{i,N}(t)\right)^2dt\right],
		\end{align*}
		where
		\begin{align*}
			d{X}^{i,N}_t &= \left(\frac{1}{N}\sum_{j=1}^NG^{N}_{i,j}{X}^{j,N}_t+Q^{i,N}(t)\right)dt +d{W}^{i,\bar{Q}^N}_t,\quad \PP^{\bar{Q}^N}-a.s.\\
			X_0^{i,N} &=x_0^{i,N}.
		\end{align*}
        Here $\bar{Q}^N:=(Q^{1,N},\ldots,Q^{N,N})$ solves the linear $N$-dimensional ODE \eqref{ODE.system}. Thus,
		\begin{align}\label{Value.aprox}
		V_p^N-J_p^N(\xi^{N,\mathrm{cont}})\leq \frac{1}{2}\int_0^T\sup_{i\in [N]}\left|Q\left(s,\frac{i}{N}\right)^2-Q^{i,N}(s)^2 \right|ds +\sup_{i\in [N]}\bigg|\EE^{\hat{Q}^N}[X^i_T]-\EE^{\bar{Q}^N}[X^i_T]\bigg|.
		\end{align}
		Next, for all $t\in [0,T]$, standard arguments give
		\begin{align*}
			\sup_{i\in[N]}\bigg|\EE^{\PP^{\hat{Q}^N}}\left[X^{i,N}_t\right]-\EE^{\PP^{\bar{Q}^N}}\left[{X}^{i,N}_t\right]\bigg| &\leq T\left(\sup_{i\in [N]}\sup_{t\in [0,T]}\big|Q^{i,N}(t)-Q(t,i/N)\big|\right)\\
            &+\|G\|_{\infty}\int_0^t\sup_{i\in[N]}\bigg|\EE^{\PP^{\hat{Q}^N}}\left[X^{i,N}_s\right]-\EE^{\PP^{\bar{Q}^N}}\left[{X}^{i,N}_s\right]\bigg|ds.
            \end{align*}
		Using Grönwall's inequality, there exists a constant $C_1>0$ such that
        \begin{align*}			\sup_{t\in [0,T]}\sup_{i\in[N]}\Big|\EE^{\PP^{\hat{Q}^N}}\left[X^{i,N}_t\right]-\EE^{\PP^{\bar{Q}^N}}\left[X^{i,N}_t\right] \Big|\leq C_1\sup_{i\in [N]}\sup_{t\in [0,T]}\big|Q^{i,N}(t)-Q(t,i/N)\big|.
        \end{align*}
        Moreover, applying ODE estimates, we have
         \begin{align*}			\sup_{i\in[N]}\Big|Q(t,i/N)^2-\left(Q^{i,N}(t)\right)^2 \Big| \leq 2e^{\|G\|_{\infty}T}\sup_{i\in [N]}\sup_{t\in [0,T]}\big|Q^{i,N}(t)-Q(t,i/N)\big|.
        \end{align*}
       Substituting the previous two inequalities into \eqref{Value.aprox} and using Theorem \ref{lemmaQ}, there exists a constant $C>0$ such that
        $$V_p^N-J^N_p(\xi^{N,\mathrm{cont}})\leq\frac{C}{N}.$$
       This proves the theorem.
		\qed
		\end{proof}
The table below summarizes the convergence rates shown in this section.
\begin{table}[htbp]
	\centering
	\small
	\caption{Summary of convergence rates in Subsection~\ref{subsectionconvergence}.}
	\label{tab:rates-summary}
	\begin{tabular}{p{0.23\textwidth}p{0.4\textwidth}p{0.1\textwidth}p{0.1\textwidth}}
		
		\toprule
		Quantity & Comparison & Rate & Reference \\
		\midrule
		Effort processes & $\max_{i\in[N]}\sup_{t\in[0,T]}|Q(t,i/N)-Q^{i,N}(t)|$ & $\mathcal{O}(N^{-1})$ & Theorem~\ref{lemmaQ} \\
		Principal value & $\mathbb{E}^{\lambda}[(V_p^N-V_p)^2]$ & $\mathcal{O}(N^{-1})$ & Theorem~\ref{conv.Value} \\
		Optimal-contract law & $\sup_{i\in[N]}\mathcal{W}_2\big(K_G(i/N,\cdot),\PP^{\bar{Q}^N}\circ(\xi^{i,N,*})^{-1}\big)$ & $\mathcal{O}(N^{-1/2})$ & Theorem~\ref{conv.Value} \\
		Scalable optimal contract loss & $V_p^N-J_p^N(\xi^{N,\mathrm{cont}})$ & $\mathcal{O}(N^{-1})$ & Theorem~\ref{nearoptimal} \\
		\bottomrule
	\end{tabular}
\end{table}

\subsubsection{Numerical convergence of the finite-agent approximation.}
We close with a numerical check of the convergence results. For each example, we solve the corresponding
$N$-agent problem for different values of $N$. Let $\bar{Q}^N$ denote the solution to the $N$-agent problem, and let $V^N_p$ be the associated principal value. We compare these quantities with the continuum solution $Q$ and $V_p$, respectively. Since the optimal action is given by $\alpha^*=Q$, the convergence of the optimal efforts is measured by the discrete $L^2$-error
\begin{align*}
    \|Q-\bar{Q}^N\|_{L^2} :=  \sqrt{\frac{T}{NM}\sum_{i=1}^N\sum_{j=1}^M\big|Q(t_j,i/N)-Q^{i,N}(t_j)\big|^2},
\end{align*}
where $t_j = \frac{(j-1)T}{M}$, $j\in [M]$, and $M$ is the length of the time discretization.

Convergence of the principal value is measured by the absolute error $\big| V_p^N -V_p\big|$.
\begin{figure}[H]
	\centering
	\includegraphics[width=0.90\textwidth]{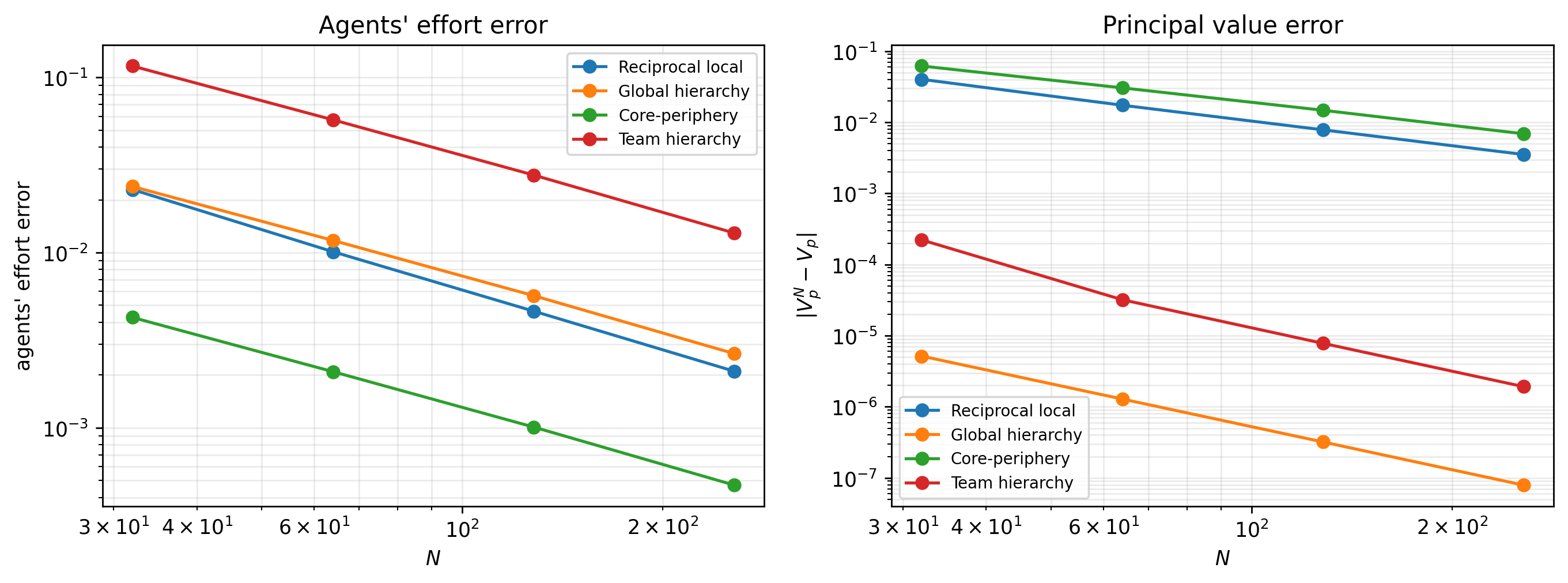}
	\caption{Convergence of the agents' optimal efforts and principal value across the four normalized interaction-function examples.}
	\label{fig:convergence-actions-value}
\end{figure}
\section{Conclusion}

This paper studies optimal contracting in large populations of heterogeneous agents whose outputs are linked through network spillovers. We model these spillovers by an interaction function and show that, in a linear-quadratic continuous-time setting, the optimal contract can be characterized explicitly in both the finite-agent problem and its continuum limit. The optimal incentive assigned to an agent is determined by the solution of a linear backward equation, and this solution has a direct economic interpretation: it is both the slope of the performance-based component of the contract and the equilibrium effort induced from that agent. Thus, the contract rewards agents according to their position in the interaction structure, with steeper incentives assigned to agents whose effort generates larger spillover effects.

The continuum formulation provides a tractable approximation to otherwise high-dimensional finite-agent contracting problems. By encoding heterogeneity through the type variable and the joint law of type and output, we reduce the continuum principal-agent problem to a mean-field control problem and obtain an explicit optimal contract. We also show that the continuum solution is stable with respect to perturbations of the interaction function. This stability result is important for applications in which the principal may estimate the interaction structure only approximately.

A central implication of the analysis is that the continuum contract can be used as a scalable contract for large finite populations. Evaluating the continuum incentive rule on a finite grid of agent types yields admissible finite-agent contracts that satisfy the agents' reservation constraints, induce a unique Nash equilibrium, and achieve the finite-agent principal's value up to an error of order $1/N$. The numerical examples illustrate this approximation property and show how different interaction structures, such as reciprocal local interactions, global hierarchies, core-periphery networks, and team hierarchies, affect equilibrium efforts, compensation, and the principal's value.

Several extensions remain natural directions for future work. One direction is to study more general preferences, including risk-averse principals or agents, nonlinear production technologies, and richer cost structures. Another is to allow for common noise or aggregate shocks, which would make the interaction between individual incentives and population-level uncertainty more explicit. It would also be useful to study settings in which the interaction function is learned from data or evolves over time. These extensions would further connect the tractable continuum approach developed here with contracting problems in large, dynamic, and partially observed economic networks.
\bibliographystyle{alpha}
\bibliography{main}

\end{document}